\documentclass[journal]{IEEEtran}

\usepackage{booktabs}
\usepackage{xcolor}
\usepackage{color}
\usepackage{bbding}
\usepackage{multirow}
\usepackage{amsthm}
\usepackage{amsmath}
\usepackage{pifont}
\usepackage{amssymb}
\usepackage{yhmath}
\usepackage{graphicx}
\usepackage[linesnumbered,ruled,vlined]{algorithm2e}
\SetKwRepeat{Do}{do}{while}
\usepackage{algorithmicx}
\SetAlFnt{\footnotesize}

\usepackage{bm}
\usepackage{setspace}
\usepackage{tabularx}
\usepackage{threeparttable}
\usepackage[caption=false,font=footnotesize]{subfig}
\graphicspath{{figs/}}
\usepackage{cite} 
\usepackage{tikz}
\usepackage[utf8]{inputenc}
\newtheorem{definition}{Definition}
\newtheorem{theorem}{Theorem}
\usepackage{caption}
\usepackage{mathrsfs}
\usepackage{textcomp}
\usepackage{graphicx}
\usepackage{float} 
\usepackage{subfig}
\usepackage{booktabs,makecell,multirow}
\usepackage{titlesec}
\usepackage{xspace}
\renewcommand\thesection{\arabic{section}}

\renewcommand{\thesubsection}{\Alph{subsection}.}

\renewcommand{\thesubsubsection}{\arabic{subsubsection})}

\titleformat{\section}
  {\normalfont\Large\bfseries}
  {\thesection}{1em}{}

\titleformat{\subsection}
  {\normalfont\normalsize\bfseries}
  {\thesubsection}{1em}{}

\titleformat{\subsubsection}
  {\normalfont\normalsize\bfseries}
  {\thesubsubsection}{1em}{\normalfont}
\usepackage{hyperref}

\newcommand{\TrinityI}{Trinity-\uppercase\expandafter{\romannumeral1}\xspace}
\newcommand{\TrinityII}{Trinity-\uppercase\expandafter{\romannumeral2}\xspace}

\ifCLASSINFOpdf
\else
\fi
\begin{document}
%
\title{Trinity: A Scalable and Forward-Secure DSSE for Spatio-Temporal Range Query}
%
%
%
\author{Zhijun~Li, 
        Kuizhi~Liu, 
        Minghui~Xu,~\IEEEmembership{Member,~IEEE,}
        Xiangyu~Wang,~\IEEEmembership{Member,~IEEE,}
        Yinbin~Miao,~\IEEEmembership{Member,~IEEE,}
        Jianfeng~Ma,~\IEEEmembership{Member,~IEEE,}
        and
        Xiuzhen~Cheng,~\IEEEmembership{Fellow,~IEEE,}
     
\thanks{Zhijun Li, Minghui Xu,  and Xiuzhen Cheng are with the School of Computer Science and Technology, Shandong University, China, Qingdao 266237, China. (e-mail: richikun2014@gmail.com, mhxu@sdu.edu.cn,  xzcheng@sdu.edu.cn).}
\thanks{Kuizhi Liu, Xiangyu Wang, Yinbing Miao, and Jianfeng Ma are with the School of Cyber Engineering, Xidian University, Xi'an 710071, China. (e-mail: ighxiy@163.com, wangxiangyu01@xidian.edu.cn, ybmiao@xidian.edu.cn, jfma@mail.xidian.edu.cn).}
\thanks{Corresponding author: Minghui Xu}

}

%
%

\markboth{IEEE TRANSACTIONS ON INFORMATION FORENSICS AND SECURITY}%
{Li \MakeLowercase{\textit{$et\ al.$}}: Bare Demo of IEEEtran.cls for IEEE Journals}
%



\maketitle

\begin{abstract}
Cloud-based outsourced Location-based services have profound impacts on various aspects of people's lives but bring security concerns. Existing spatio-temporal data secure retrieval schemes have significant shortcomings regarding dynamic updates, either compromising privacy through leakage during updates (forward insecurity) or incurring excessively high update costs that hinder practical application. 
Under these circumstances, we first propose a basic filter-based spatio-temporal range query scheme \TrinityI that supports low-cost dynamic updates and automatic expansion. Furthermore, to improve security, reduce storage cost, and false positives, we propose a forward secure and verifiable scheme \TrinityII that simultaneously minimizes storage overhead. A formal security analysis proves that \TrinityI and \TrinityII are Indistinguishable under Selective Chosen-Plaintext Attack (IND-SCPA). Finally, extensive experiments demonstrate that our design \TrinityII significantly reduces storage requirements by 80\%, enables data retrieval at the 1 million-record level in just 0.01 seconds, and achieves 10 $\times$ update efficiency than state-of-art.
 
\end{abstract}


%
\IEEEpeerreviewmaketitle

\section{Introduction}
%
%
%
%
\IEEEPARstart {T}{he} popularity of smart devices has accelerated Location-Based Services (LBS), which retrieve and sort information based on user location and requests, then provide recommendations to users. 
In traditional spatial LBS, users can ask if a stationary object sits within a specific area. However, how do users track the movement patterns of dynamic objects such as vehicles? That is where spatio-temporal LBS comes in, stepping up from mere spatial snapshots to track and query objects across both space and time. 

For spatio-temporal LBS service providers, outsourcing spatio-temporal data to the cloud is a practical method way of improving the service quality and reducing local operating costs. However, outsourcing data to third parties inevitably brings security and privacy issues. For example, in 2018, fitness tracking app Strava published a global heat map that visualized the activity of its users, inadvertently revealing the locations and activity time of military personnel at sensitive facilities around the world. Therefore, Dynamic Searchable Symmetric Encryption (DSSE) \cite{stefanov2014practical} has become a treatment to this concern, as it allows clients with specific keys to query and update ciphertext without revealing sensitive data. However, almost all DSSE schemes allow for certain information leakage in exchange for efficiency \cite{cao2019protecting,huang2021netr,zhu2021privacy,wang2022quickn,guo2022search,yang2022lightweight,miao2023efficient}. Therefore, attackers can potentially disclose the contents of past queries by inserting new documents, as the server is capable of recognizing matches between these newly inserted documents and previous search queries \cite{zhang2016all}. To avoid such leakage, several studies \cite{kermanshahi2020geometric,li2021secure,wang2022forward,li2023enabling} have focused on introducing forward security to prevent such leaks that occur during DSSE updates.

When applying DSSE to spatio-temporal scenarios, scalability becomes an critical issue. Due to the spatio-temporal nature of the data, the size of the database grows continuously, making scalability a critical concern in system design. 
The scalability of DSSE schemes largely depends on the index structures, including tree-based index structures \cite{kermanshahi2020geometric,etemad2018efficient} and linear index structures \cite{wang2017fastgeo,dou2024dynamic}. Tree-based spatio-temporal range query schemes inherently face scalability challenges \cite{li2019efficient,zheng2020practical}. When data volume grows, we have to frequently expand tree structures by adjusting tree height and reorganizing nodes, which poses scalability limitations \cite{wang2014maple}. 

Compared to tree-based schemes, filter-based ones are more scalable based on an array-based index structure \cite{miao2023efficient}. Furthermore, the filters demonstrate superior space efficiency compared to the R-trees \cite{cui2019geo}. These advantages have led to the increasing adoption of filters in spatio-temporal range query schemes. Filter-based DSSE for spatio-temporal range query work by transforming high-dimensional spatio-temporal data into one-dimensional representations stored in filters. By leveraging a filter, which maps elements to multiple locations using hash functions, queries can be executed through simple bit operations with low computational complexity \cite{cui2019geo}. This approach effectively utilizes the filters' fast lookup capability to efficiently determine element existence and process spatio-temporal range queries. Many DSSE schemes are developed based on bloom filter  \cite{cui2019geo,wang2021enabling,wang2022forward,zhang2022efficient,miao2023efficient}, there still remain challenges that need to be addressed further.  

\subsection{Challenges} 

\textbf{Challenge~I: How to construct a filter that enables flexible and efficient element deletion.} Existing range query schemes \cite{cui2019geo,wang2021enabling,zhu2021privacy,zhang2022efficient} based on the bloom filter do not support deletion, because the bloom filter uses a single-digit set to represent the existence of an element. Each bit of the bloom filter is not exclusive, and multiple elements may share the same bit \cite{almeida2007scalable}.
Therefore, if a bit corresponding to an element is deleted, it may negatively affect other elements that share the same bit \cite{pandey2021vector}. For example, if the two elements “seafood" and “Japanese cuisine" happen to overlap in some bit positions. Setting the bit position corresponding to the “seafood" element to 0 for deletion can cause the “Japanese cuisine" element to be accidentally deleted. Although there are some none-filter DSSE schemes \cite{kermanshahi2020geometric,wang2022forward} that support deletion, our experiment results show that their efficiency remains a concern and they do not apply to filter-based schemes.

\textbf{Challenge~II: How can the scalability of forward-secure filters be improved.} Existing forward-secure range query filters \cite{wang2021enabling,zhu2021privacy,wang2022forward} face scalability problems. As the amount of spatio-temporal data increases over time, the probability of hash collisions increases, which affects the query efficiency and false positive rate \cite{zhang2022efficient}. Especially when the number of data entries in the database reaches its maximum capacity, performing the add operation will cause the database to crash or rebuild. 
Moreover, the bit group size of a filter is static and cannot be modified once set \cite{mullin1983second}. Expanding the size would require remapping existing data, which is impractical \cite{almeida2007scalable}. While initially overprovisioning filter capacity might seem viable, it results in substantial wasted space. Additionally, assuming consistent data input is unrealistic. Periods of inactivity would further amplify space inefficiency.

\textbf{Challenge~IIIs: How to overcome the trade-off between saving storage and minimizing false positives.} Filter-based schemes \cite{li2022adaptively,li2023vrfms,tong2023verifiable} need free space to keep a low false positive rate (FPR). For example, the FFR of widely sued bloom filters is given by $\epsilon = \left ( 1-e^{-k\cdot \frac{n}{m} } \right ) ^{k} $, where $k$ is the number of hash functions, $m$ is the length of the filter, and $n$ is the number of inserted elements \cite{fan1998summary}. 
To achieve a practical FPR of 0.01$\%$, a balance between accuracy and efficiency is sought \cite{almeida2007scalable}. This configuration inherently sets the ratio $\frac{m}{n} =\frac{k} {\ln 2}$, where $k=-\frac{\ln p }{\ln 2}$. This fixed ratio imposes a fundamental design constraint, requiring $\frac{m}{n}\geq 20$ to maintain the desired FPR. In other words, at least 19 times the number of inserted elements in free space is necessary. Overcoming this trade-off is crucial for improving filter-based DSSE.

\begin{table*}[htbp]
\centering
\caption{Comparison With Prior Works}
\label{comparisions}
\setlength{\tabcolsep}{3mm}
\begin{threeparttable}
\begin{tabular}{lccccccc}
\toprule
Schemes & \makecell{Cryptographic\\Primitive} & \makecell{Dynamic\\Update} & \makecell{Forward\\Security} & \makecell{High\\Efficiency} & \makecell{Spatio-\\Temporal} & \makecell{Scalability} & \makecell{Security \\Model}\\
\midrule

$\mathsf{DSSE}_{\mathsf{SKQ}}$\cite{wang2022forward} & ASHE & \checkmark & \checkmark & \checkmark & \ding{53}& \ding{53} & IND-CPA \\
SKSE-\uppercase\expandafter{\romannumeral 2}\cite{wang2021enabling} & HVE & \ding{53} & \ding{53} & \checkmark & \ding{53}& \ding{53} & IND-SCPA \\
GRS-\uppercase\expandafter{\romannumeral 2}\cite{kermanshahi2020geometric} & ASHE & \checkmark & \checkmark & \ding{53} & \ding{53}& \ding{53} & IND-CPA \\
Trinity-\uppercase\expandafter{\romannumeral 1} & SHVE & \checkmark & \ding{53} & \checkmark & \checkmark & \ding{53}& IND-SCPA \\
Trinity-\uppercase\expandafter{\romannumeral 2} & SHVE & \checkmark & \checkmark & \checkmark & \checkmark & \checkmark &IND-SCPA \\
\bottomrule
\end{tabular}
\begin{tablenotes}
\footnotesize
\item \textbf{Notes.} ASHE stands for Additive Symmetric Homomorphic Encryption. ASPE stands for Asymmetric Scalar Product-Preservation Encryption. HVE stands for Hidden Vector Encryption. SHVE stands for Symmetric-key Hidden Vector Encryption.
\end{tablenotes}
\end{threeparttable}
\end{table*}


\textbf{Our Contributions}. We introduce two novel DSSE schemes tailored for spatio-temporal range queries.
\begin{itemize}
    \item \textbf{\TrinityI}. A fundamental DSSE construction that supports dynamic and efficient element updates, including addition and deletion. \TrinityI offers a performance boost over state-of-the-art methods, achieving up to a 20$\times$ speedup. And \TrinityI adaptively expands its index structure, so we can keep a negligible false positive rate and low latency on search and update for the filter. Besides, our fundamental DSSE construction is secure against IND-SCPA. 
    \item \textbf{\TrinityII}. Building upon \TrinityI, this scheme provides additional high accuracy, and forward security. \TrinityII offers a verification function that strikes a balance between storage efficiency and low false positives. It reduces storage requirements by 80\% while enabling data retrieval at the 1 million-record level in just 0.01 seconds, and offers enhanced forward security against file injection attacks. \TrinityII still offers a performance boost over state-of-the-art methods, achieving up to a 10$\times$ speedup on update latency.
\end{itemize}

\hfill 
\section{Related Work}
In this section, we introduce some related work. To ensure a more impartial and lucid presentation, we list Trinity with previous works in Table~\ref{comparisions}.
\par  \textbf{DSSE.} DSSE is a cryptographic technique that enables clients to securely outsource encrypted databases to servers while still maintaining the ability to perform search and update operations on encrypted data \cite{stefanov2014practical,cash2015leakage,kim2017forward, 10621113, etemad2018efficient,zuo2020forward,wang2022forward}. While DSSE provides efficient access to encrypted databases, it inevitably incurs some information leakage, particularly during update operations, where sensitive data can potentially be exposed, thus compromising privacy. The concept of forward security, introduced by Stefanov $et\ al.$ \cite{stefanov2014practical}, aims to address this issue by ensuring that updates to the encrypted database do not reveal information about previous queries or data. Zhang $et\ al.$ \cite{zhang2016all} further highlighted the vulnerability of DSSE to leakage attack, specifically file injection tactics, where an attacker strategically inserts a limited number of documents into the encrypted database to deduce search queries. Following this finding, Bost \cite{bost2016ovarphiovarsigma} provided a formal definition of forward security and proposed a concrete scheme that incorporates forward security in DSSE. Since then, several forward-secure DSSE schemes have been proposed \cite{etemad2018efficient,song2018forward,li2021towards,guo2023forward}, with a focus on improving security, efficiency, and functional diversity.



\par \textbf{Filter-based DSSE.} 
The common drawback of filter-based search DSSE schemes \cite{wang2021enabling, zhu2021privacy,li2022adaptively, zhang2022efficient, miao2023efficient, li2023vrfms, tong2023verifiable}, particularly bloom filter-based schemes, is that they cannot support deletion. These filters are designed to efficiently test membership in a set while allowing for false positive results, meaning they can indicate that an element is present in the set when it is not. This characteristic is inherent in their probabilistic nature, which prioritizes space efficiency over absolute accuracy. Wang $et\ al.$ \cite{wang2015circular} introduced a pioneering privacy-preserving circular range search scheme; however, the search token size scales with the square of the search radius R, posing scalability challenges. To address these issues, Wang $et\ al.$ \cite{wang2021enabling} ingeniously leveraged the Hilbert curve, effectively reducing queries of the two-dimensional spatial range to one-dimensional searches, thus significantly enhancing retrieval speed. Although this approach represents a notable improvement, it still lacks support for dynamic updates and is limited to a fixed data scale. Similarly, the non-scalable structure employed in \cite{zhang2022efficient, miao2023efficient} requires significant space overhead to maintain a low False Positive Rate (FPR), further highlighting the limitations of these methods.

\textbf{None-filter DSSE.} Balancing efficiency and security has consistently presented a challenge in privacy-preserving LBS \cite{kerschbaum2014optimal,grubbs2017leakage,wang2022quickn,guo2022search}. Various effective privacy-preserving approaches have been proposed, including the use of Order-Revealing Encryption (ORE) \cite{kerschbaum2014optimal}, and a variant of Order-Preserving Encryption (OPE) \cite{wang2022quickn}. However, recent research \cite{grubbs2017leakage} has highlighted severe security vulnerabilities in property-preserving encryption schemes such as OPE and ORE. Cui $et\ al.$ \cite{cui2019geo} contributed a privacy-preserving Boolean Range Query (BRQ) solution, which was later found to be vulnerable to Ciphertext-Only Attacks (COA) \cite{li2019insecurity} due to the use of Asymmetric Scalar Product-Preservation Encryption (ASPE) \cite{wong2009secure}.
To address this security issue, Yang $et\ al.$ \cite{yang2022lightweight} proposed an enhanced ASPE scheme that achieves Indistinguishability under Chosen Plaintext Attack (IND-CPA). However, their approach compromised on search efficiency. Kermanshahi \cite{kermanshahi2020geometric} proposed a forward-secure DSSE for spatial range query, but the efficiency of their solution is unacceptable. Wang $et\ al.$ \cite{wang2021enabling} then improved search speed by incorporating a Bloom filter hierarchical tree, but their solution was limited to single-user scenarios due to its symmetric key setting and focused solely on spatial queries, neglecting the temporal dimension.  Li $et\ al.$ \cite{li2023enabling} developed an efficient privacy-preserving spatio-temporal LBS scheme capable of retrieving million-level data in milliseconds.

\section{Preliminaries}
Here, we introduce some crypto primitives and data structures widely used in Trinity including Hilbert curve, bloom filter, and SHVE.
\subsection{Hilbert Curve}
A Hilbert curve is a type of continuous factal space-filling curve that fills a $d$-dimensional area \cite{hilbert1935stetige}. For a $d$-dimensional space, each dimension is uniformly divided into $2^h$ segments, where $h$ denotes the order of the Hilbert curve. The Hilbert curve is constructed by recursively dividing an area into $2^d$ smaller areas and connecting the centers of these smaller areas in a specific order. With the Hilbert curve, any $d$-dimensional spatial range query can be converted into a range query in one-dimensional space of $2^{dh}$ consecutive areas, each represented by a $dh$-bit value. 




\subsection{Quotient Filter}
\par A quotient filter is a variant of the Bloom filter. It can quickly detect whether an element is in a set, indicating that the data definitely do not exist or may possibly exist. The quotient filter employs a space-efficient hash table mechanism, dividing the $p$-bit fingerprint of an element into two distinct segments. The first segment, known as the quotient, consists of the $q$ most significant bits and is used to rapidly determine the “target position" of the element within the filter. The second segment, referred to as the remainder, encapsulates the $r = p - q$ least significant bits and serves to differentiate elements with the same quotient. The original slot where the hash output point's quotient resides is called the \texttt{canonical slot}. A sequence of consecutive slots containing remainders with the same quotient is termed a \texttt{run}. A \texttt{cluster} in a quotient filter is a contiguous sequence of \texttt{run}s, commencing with the first \texttt{run} whose initial fingerprint occupies its \texttt{canonical slot}, and extending until an unoccupied slot is encountered or the another \texttt{run} occupies its \texttt{canonical slot} is identified. Each slot contains a 3-bit counter.
\begin{itemize}
	\item \textbf {Is\_ occupied}. Set to 1 if a slot is occupied correctly, that is, a quotient value corresponds to the slot index.
	\item \textbf {Is\_ continuation}. Set to 1 when a slot is not the start of the \texttt{run}. 
	\item \textbf {Is\_ shifted}. Set to 1 when this slot is not the start of the \texttt{cluster}. It happens when there is an offset between the position where the remainder is stored and the index represented by the quotient associated with the remainder.
\end{itemize}

\begin{figure}[!ht]
	\centering
	\includegraphics[scale=0.3]{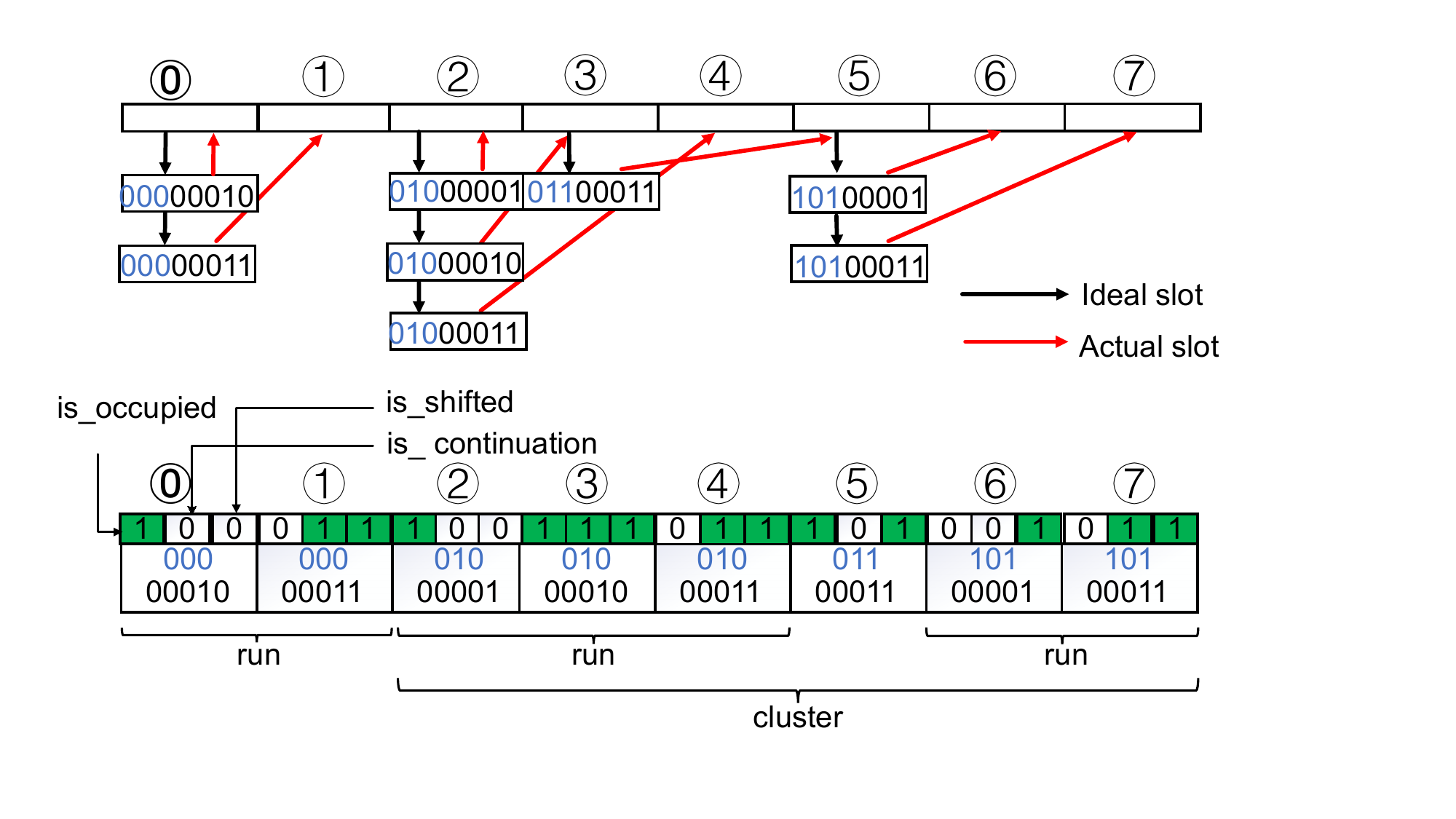}
	\caption{An example of quotient filter}
	\label{quotient filter}
\end{figure}

We provide an example of the quotient filter in Fig.~\ref{quotient filter}. The upper portion depicts a simplified view, where eight fingerprints are categorized into buckets based on their quotient (blue numbers). For example, the first fingerprint 000$\big|$00010, with the quotient of 000, is placed in slot \textcircled{0}. The second fingerprint 000$\big|$00011 ideally belongs to slot \textcircled{0}, but slot \textcircled{0} is occupied, so we shift it right to slot \textcircled{1} (the actual slot of the second fingerprint).
The lower portion shows the actual storage mechanism. Each bucket holds a 3-bit counter and a remainder. In this example, 000$\big|$00010, and 000$\big|$00011 share the same quotient, forming a \texttt{run}. The eight fingerprints finally constitute three \texttt{run}s and one \texttt{cluster}.


\subsection{Symmetric-key Hidden Vector Encryption}

Symmetric-key Hidden Vector Encryption (SHVE) is a lightweight crypto primitive based on Hidden Vector Encryption (HVE). Like HVE, it supports conjunctive, equality , and subset queries on encrypted data. Let $\Gamma$ be a finite field and ‘$\ast$’ be a wildcard symbol not belong to $\Gamma$, $\Gamma_{\ast}=\Gamma \bigcup \{\ast \}$. The SHVE is defined in the subsequent algorithm:
\begin{itemize}
	\item \textbf{SHVE.Setup$(\{0,1\}^{\lambda })$}: On input a security parameter ${\lambda }$ , it returns a master secret key $msk$ and the message space $\mathcal{M}$.

\item \textbf{SHVE.KeyGen$\left(msk,{\bf V}\in \Gamma_{*}^{d}\right)$}: On input a master secret key $msk$ and a predicate vector $\displaystyle\mathbf{v}=~\{v_{1},\ldots,v_{d}\}~\in~\Gamma_{*}^{d}$, it returns the $s$ as a decryption key.
\item \textbf{SHVE.Enc$(msk,\mu \in \mathcal{M},\mathbf{ind} \in \Gamma^{d})$}: On input an index vector $ind$, a master key $msk$ and a message $\mu$, it returns the ciphertext $c$ associated with $(\mathbf{ind},\mu)$.

\item \textbf{SHVE.Query$\left(\mathbf{c},\mathbf{s}\right)$}: On input a ciphertext $c$ associated with $(\mathbf{ind},\mu)$ and a decryption key $s$ corresponding to the predicate vector $v$, it returns $\mu$ if ${\cal P}_{\mathrm{v}}^{\mathrm{SHVE}}({\bf x})=1$; else returns null.
\end{itemize}

\section{Problem Formulation}
In this section, we illustrate the system model and the threat model considered in this paper.
\subsection{System Model} Trinity involves three entities: Cloud Server (CS), Data Owner (DO) and Data Users (DU).
\begin{itemize}
	\item \textsf{Cloud server}. The CS is assumed to have sufficient computing and storage capabilities. The primary responsibilities of the CS include storing the ciphertext and secure index provided by the DO. Upon receiving a search token from a DU, the CS processes the token and returns the results to the DU.
	\item \textsf{Data owner}. The DO encrypts the files and generates the corresponding security index, which is uploaded to the CS. Additionally, the DO generates a secret key and securely distributes it to the authorized DU through a secure channel.  
	\item \textsf{Data users}. The DU obtains the key from the DO and creates search tokens corresponding to the spatio-temporal objects being queried. These tokens are then sent to the CS. DU receives the results from the CS after the search is completed.
	
\end{itemize}
A spatio-temporal database ${\bf DB}$ is defined as: ${\bf DB}=\{{\bf p}_{i},\ {\bf ind}_{i}\}_{i=1}^{N}$, ${\bf p}_{i}=(x_i,y_i,z_i)$ is a space-time node, ${\bf ind}_{i}$ is a file identifier, and $N$ is the size of the quotient filter.
\begin{definition}
\label{3}
For an encrypted spatio-temporal database ${\bf EDB}$ and a search token ST, Trinity is to search a subset $\{{\bf ind}_{i}^{\ast}\}_{i=1}^{c}$ such that $\forall 1\leq i\leq c$, $\mathcal{P}_{i}^{*}\in\mathcal{R}_{\mathcal{Q}}$, where $\mathcal{P}_{i}^{*}$ represents the minimum set of prefix elements that collectively cover the
entire range query $\mathcal{R}_{\mathcal{Q}}$.

\end{definition}
\begin{definition}
\label{4}
The Trinity is a DSSE scheme consisting three algorithms described as follows.
\begin{itemize}
	\item \textup{\textbf{Setup}$(\{0,1\}^{\lambda })\rightarrow(K_{\Sigma},\sigma;\mathrm{EDB})$}: On input a security parameter ${\lambda }$, it returns a secret key set $K_{\Sigma}$, a state $\sigma$ and the encrypted database $\mathrm{EDB}$.

\item \textup{\textbf{Search}} $(\mathrm{Q},K_{\Sigma},\sigma;\mathrm{EDB})\rightarrow \mathrm{R}$: On input a query $\mathrm{Q}$, a secret key set $K_{\Sigma}$, a state $\sigma$ and the encrypted database $\mathrm{EDB}$, the client sends a search request to the server, and the server returns the result $\mathrm{R}$ after searching all over the $\mathrm{EDB}$.
\item \textup{\textbf{Update}}$(K_{\Sigma},O,up,\sigma;\mathrm{EDB})\rightarrow (K_{\Sigma},\sigma^{\prime};\mathrm{EDB}^{\prime})$: On input $(K_{\Sigma},O,up,\sigma;\mathrm{EDB})$, where $O=\{p,ind\}$ and $up\in\{add,del\}$. $up$, $add$ and $del$ denote update, addition, and deletion, respectively. The server adds or deletes the object $O$ from $\mathrm{EDB}$.
\end{itemize}

\end{definition}

\subsection{Threat Model} 

We assume that the Data Owner (DO) and Data User (DU) are trustworthy. The Cloud Server (CS), however, is considered honest-but-curious, meaning it faithfully executes the protocol but may attempt to extract information from encrypted data. 
To ensure the security of the DSSE system (e.g., Trinity), the adversary, represented as $\mathcal{A}$, should not be able to gain any meaningful information from the encrypted databases or queries. This security property is formalized through a real-world vs. ideal-world game-based approach. 
Let $\mathcal{L} = \{\mathcal{L}^{Stp}, \mathcal{L}^{Srch}, \mathcal{L}^{Updt}\}$ denote the leakage function that captures the information an adversary can potentially obtain from the \textbf{Setup}, \textbf{Search}, and \textbf{Update} phases. The formal definitions of these games are as follows:

\begin{itemize}
 \item The real game $Game_{\mathcal{A}}^{\mathcal{R}}(\lambda)$: $\mathcal{A}$ selects a database DB and generates the EDB via \textup{\textbf{Setup}}.  Subsequently, $\mathcal{A}$ adaptively runs \textup{\textbf{Search}} or \textup{\textbf{Update}}. Throughout the experiment, $\mathcal{A}$ maintains full observational access to the real operational transcripts. Upon conclusion of the adaptive query phase, $\mathcal{A}$ outputs $b \in \lbrace 0,1\rbrace$.
 \item The ideal game $Game_{\mathcal{A}}^{\mathcal{S}}(\lambda)$: $\mathcal{A}$ selects a database DB and generates the EDB via simulator $\mathcal {S}(\mathcal {L}^{{Stp}}(\mathsf {DB}))$.  Subsequently, $\mathcal{A}$ adaptively runs $\mathcal {S}(\mathcal {L}^{{Srch}})$ or $\mathcal {S}(\mathcal {L}^{{Updt}})$. Throughout the experiment, $\mathcal{A}$ maintains full observational access to the simulated operational transcripts. Upon conclusion of the adaptive query phase, $\mathcal{A}$ outputs $b \in \lbrace 0,1\rbrace$.
 \par If there is an efficient simulator $\mathcal{S}$ such that:

\end{itemize}
If $\mathcal{A}$ cannot distinguish between the real-world game $Game_{\mathcal{A}}^{\mathcal{R}}(\lambda)$ and the ideal-world game $Game_{\mathcal{A}}^{\mathcal{S}}(\lambda)$, then the system is considered secure, meaning no information beyond the specified leakage function is leaked. 

\begin{definition}
\label{def1}
A DSSE is $\mathcal {L}$-adaptively secure \cite{curtmola2006searchable} if for any probabilistic polynomial-time (PPT) adversary $\mathcal{A}$, there is an efficient simulator $\mathcal{S}$:
 \begin{equation}
    \nonumber
  \vert Pr[Game_{\mathcal{A}}^{\mathcal{R}}(\lambda)\!\!=\!\!1]-Pr[Game_{\mathcal{A}}^{\mathcal{S}}(\lambda)=1]\vert \leq \mathsf{negl}(\lambda). 
    \end{equation}
\end{definition}

Besides, forward security plays a pivotal role in safeguarding DSSE schemes from leakage-abuse attacks. The DSSE scheme is considered 'forward-secure' if there is no connection between an update of encrypted data and any previously performed search results. Formally speaking, 
\begin{definition}
\label{def2}
\par  A $\mathcal {L}$-adaptively secure DSSE is forward-secure \cite{bost2016ovarphiovarsigma} if the update leakage function $\mathcal {L}^{Updt}$ is defined as follows,
	\begin{equation}
    \nonumber
  \mathcal {L}^{Updt}(op,in=(p,ind))=\mathcal {L}^{\prime }(up,in=(ind,c)) ,
    \end{equation}
   
   where $op$ denotes the operation like addition or deletion, $in$ denotes the input, $ind$ represents the document identifier, and $c$ is the number of update files.
 \end{definition}

\subsection{Design Goals} 
The design goals of Trinity are described as follows.
\begin{itemize}
	\item \textbf{Dynamic.} The proposed scheme is designed to be dynamically configurable, enabling document addition and deletion operations.
	\item \textbf{Update-efficient.} The proposed scheme aims to achieve more efficient updates compared to existing forward-secure schemes.
    
    \item \textbf{Scalable.} The proposed scheme aims to be more scalable, which means it can efficiently expand its capacity for a continuous stream of data.  \item \textbf{Verifiable.} The proposed scheme aims to be verifiable, which means there would be no false positives in results and no storage waste to keep false positives minimized.  
        \item \textbf{Privacy- preserving.} The proposed scheme aims to be $\mathcal {L}$-adaptively secure and forward-secure. The system should safeguard sensitive information such as file collections, indexes, and background information of keywords from unauthorized access. Furthermore, our scheme adheres to forward security principles, ensuring that the CS remains oblivious to any association between recent updates and prior search results.
\end{itemize}

\section{Trinity Schemes}

In this section, we first introduce a scalable, update-efficient spatio-temporal range search scheme, \TrinityI, that builds on quotient filter, Hilbert curve and SHVE. Subsequently, we propose a forward-secure, and verifiable scheme Trinity-\uppercase\expandafter{\romannumeral2}. Trinity-\uppercase\expandafter{\romannumeral1} is faster than Trinity-\uppercase\expandafter{\romannumeral2} in search and update latency, but it wastes more storage and lacks forward security.

\subsection{\TrinityI: Basic Trinity Construction} 
\subsubsection{Technique Overview}
We treat spatio-temporal data as three-dimensional data and use the Hilbert curve to reduce it to a one-dimensional form. This technique effectively translates spatio-temporal range queries in a multi-dimensional space into multiple one-dimensional range queries. We denote the Hilbert curve of a point p or a range R as $\mathcal{H}(p)$ and $\mathcal{H}(R)$, respectively. Although it is feasible to encode all elements within a specified range into a QF and test for the presence of an encoded space-time node, the QF's size scales linearly with the number of elements. Notably, range queries typically encompass significantly more elements than individual data objects. To optimize QF size, we must minimize the number of elements involved in range queries. To address this challenge, we employ a prefix membership verification technique introduced by Liu $et\ al.$ \cite{liu2010privacy}.

 Previously mentioned bloom filter-based schemes \cite{li2022adaptively, zhang2022efficient, miao2023efficient, li2023vrfms, tong2023verifiable} suffer from several challenges: deletion capability, scalable structure, and minimizing false positives. In response to these limitations, our approach leverages the quotient filter, and scalable bloom filter technology to resolve those challenges. The quotient filter employs a fingerprint-based storage mechanism, where each element is represented by a unique fingerprint comprising quotient and remainder components, thereby facilitating precise element identification and efficient deletion operations. 

The system's performance inversely correlates with data volume: as data volume increases, both insertion and retrieval operations become less efficient, while false positive rates increase. So we adjust the structure of the quotient filter to make it scalable, i.e. when the amount of inserted data approaches a certain threshold (affecting performance or false positive rate is high), it automatically expands dynamically. We simply borrow one bit from the remainder into the quotient instead of rehashing all elements for the expanding. 
 
 To ascertain the subset relationship between encrypted vector representations, we leverage SHVE for performing privacy-preserving set membership queries within a cryptographically secured computational domain \cite{liang2023privacy}. Specifically, our scheme uses multi-threaded computing technology to speed up SHVE.

\subsubsection{Details of Basic Trinity Construction} 
In this section, we will briefly introduce \textbf{Setup} and \textbf{Search}, focusing on dynamic \textbf{Update}. Here dynamic \textbf{Update} refers to the ability to efficiently modify the data structure that supports the encryption scheme, allowing for the addition and deletion of elements without requiring a complete re-encryption of the dataset. Specially, our scheme uses the murmur hash function and multi-threaded computing SHVE.

\begin{table}
\centering
\label{Notation}
\caption{Notation Description for Trinity}
\resizebox{\linewidth}{!}{
\label{notation}
\begin{tabular}{cl}
\hline
Notations& Descriptions\\
\hline
$QF$& Quotient Filter\\
$\xi, \xi^{\prime}$& random values between (0, 1)\\
$\check{p}$& spatio-temporal object vector\\
$Q$& search query\\
$N$& size of Quotient Filter\\
$OT_{c}$& order token corresponding to the current count\\
$e_{c+1}$& salted token \\
${M_{1\_{c}}},{M_{2\_{c}}}$& salted matrices\\
$p_i$& encoded spatio-temporal data\\
$O$& space-time node\\
$ind_i$&  the i-th file identifier.\\
$H()$& hash function\\
$\mathcal{H}()$& Hilbert curve encode\\
$P()$& prefix of the nodes\\
$\mathcal{P}()$& prefix of the range\\
$e_{i}$& fingerprint of file identifier\\

${{F_{Q}}}$& quotient of fingerprint\\
${{F_{R}}}$& remainder of fingerprint\\
$ST_{non}$& Search token for non-leaf nodes\\
$ST_{L}$& Search token for leaf nodes\\
\hline
\end{tabular}
}
\end{table}
\par \textbf{Setup$(\{0,1\}^{\lambda })\rightarrow(K_{\Sigma},\sigma;\mathrm{EDB})$:} Given a security parameter $\lambda$ and a database DB, it returns a master key $msk$, $t$ hash functions $H=\{{H}_{i}(\cdot)\}_{i=1}^{t}$, and an encrypted database EDB. And DO share the secret key set $K_{\Sigma}=\{msk,H\}$ with the DU. For each space-time node $O_i=\{p_i,ind_i\}$, encoding the spatio-temporal data. \begin{figure}[!h]
	\centering
	\includegraphics[scale=0.39]{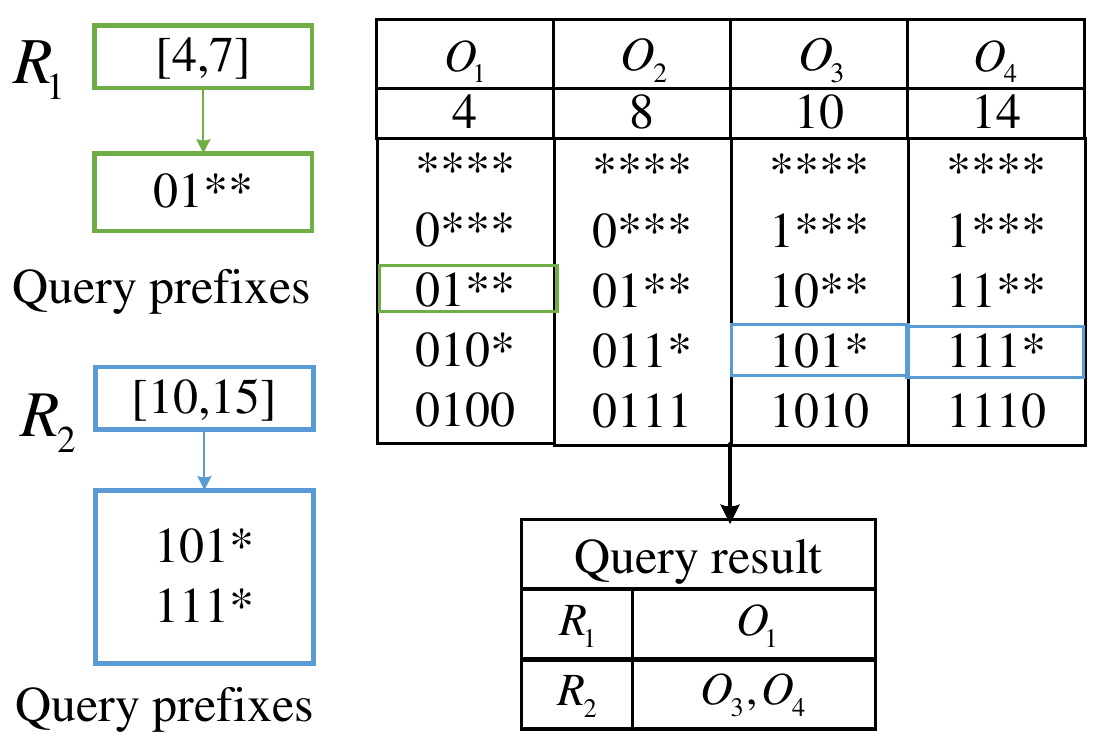}
	\caption{An example of encoded Hilbert curve with prefixes }
	\label{prefix}
\end{figure}
For a given $\omega$-bits data item $X=a_{1}~a_{2}\cdot\cdot\cdot a_{\omega}$, its prefix family is defined as a set of $\omega +1$ elements: ${P}(X)=\{a_{1}\,a_{2}\cdot\cdot\cdot a_{\omega},a_{1}\,a_{2}\cdot\cdot\cdot a_{\omega-1}{*},\cdot\cdot\cdot,a_{1}{*}\cdot\cdot\cdot{*},\ast\ast\cdot\cdot\cdot\ast\}$, where the $i$-th prefix element is $a_{1}\ \ a_{2}\cdot\cdot\cdot a_{\omega-i+1}\ *\cdot\cdot\cdot*$. Given a range $[x_{min}, y_{max} ]$, the query $\mathcal{P}\bigl(\left[x_{m i n},\,x_{m a x}\right]\bigr)$ represents the minimum set of prefix elements that collectively cover the entire range. An item X belongs to the range $[x_{min}, y_{max} ]$ if and only if the intersection of its prefix family $P(X)$ and the query prefix set $\mathcal{P}\bigl(\left[x_{m i n},\,x_{m a x}\right]\bigr)$ is non-empty: $X\in[x_{m i n},x_{m a x}] \Leftrightarrow P(X)\cap \mathcal{P}([x_{m i n},x_{m a x}])\neq \emptyset$. 
\begin{algorithm}
\caption{Setup ($1^\lambda$)}
\label{alg:Setup}
\KwIn{$(\{0,1\}^{\lambda }),\mathrm{DB})$}
\KwOut{$(K_{\Sigma},\sigma;\mathrm{EDB})$}
$\mathrm{Setup(1^{\lambda})}\rightarrow msk;$
\;
Randomly generates $H=\{{H}_{i}(\cdot)\}_{i=1}^{t}$\;
$K_{\Sigma}=\{msk,H\}$\;
 \For {{\rm space-time node} $O_i=\{p_i,ind_i\}$ }{
 $H({P}(\mathcal{H}(p_{i})))\rightarrow e_{i}$\;
 $e_{i}\rightarrow ({{F_{Q}}},{{F_{R}}})$\;

 ${\rm QF}\{H,({{F_{Q}}},{{F_{R}}})\}\rightarrow {\rm QF_i}$\;
 SHVE.Enc$(msk,\textbf{“True"},{\rm QF_i})\rightarrow \rm{C_{SHVE}}(QF_i)$\;
 $\{\rm{C_{SHVE}}({\rm QF_{i}}),ind_i\}_{i=1}^{N}\rightarrow $EDB\;

}

  \textbf{return} $(K_{\Sigma},\sigma;\mathrm{EDB})$ \;
\end{algorithm}
Fig.~\ref{prefix} demonstrates examples of spatio-temporal range queries over a spatio-temporal database utilizing our spatio-temporal data and query encoding methodology. Consider a spatio-temporal database holding four space-time nodes, each encoded via the Hilbert curve and represented as a prefix family. To retrieve objects within region $R_1$, the DU initially encodes the query range as $\{[4, 7]\}$ and subsequently generates the corresponding query prefixes as $\{ 01**\}$. Since $001***$ belongs to the prefix family of ${P}(\mathcal{H}({O}_{1})))$, we conclude that $O_1$ resides within $R_1$. To query objects within $R_2$, the DU similarly encodes the query range as $\{[10, 15]\}$ and constructs the corresponding query prefixes as $\{101*, 111*\}$. We determine that $O_3$ and $O_4$ belong to $R_2$ based on the presence of $101* \in {P}(\mathcal{H}({O}_{3})))$ and $111* \in {P}(\mathcal{H}({O}_{4})))$. We create a quotient filter and store the fingerprints ${{H}}({p}_{i})$ of the space-time node in a trinity form $({{H}}({p}_{i}),{{F_{Q}}},{{F_{R}}})$. Finally, we encrypt all the elements within the quotient filter using $msk$. A formal representation of the \textbf{Setup} phase is displayed in Algorithm~\ref{alg:Setup}.

\par\textbf{Search $(K_{\Sigma},\mathrm{Q},\sigma;\mathrm{EDB})\rightarrow \mathrm{R}$}: Given a query $\mathrm{Q}$, a secret key set $K_{\Sigma}$, a state $\sigma$, and the encrypted database $\mathrm{EDB}$, the DU sends a search request to the server. The CS then returns the result $\mathrm{R}$ after searching  through the $\mathrm{EDB}$. To search for a value in a Quotient Filter, DU first calculates the fingerprint $e_q$, and divides it as quotient $F_Q$ and remainder $F_R$, and encrypts them with $msk$. Then Du sends the search token $\rm{S_{SHVE}}({\rm QF_Q})$ to the CS. The CS checks if the target slot corresponding to ${T\_F_{Q}}$ is occupied. If it is occupied, find the start index of the \texttt{run} corresponding to the quotient $F_Q$. The CS continues the loop when the current slot is a continuation of \texttt{run}. The CS starts from the starting position $s$, compares the remainder one by one. If the remainder is equal to ${{F_{R}}}$, return true. If the remainder is greater than ${{F_{R}}}$, return false. If there is no matched remainder in \texttt{run}, return false. A formal representation of the \textbf{Search} phase is displayed in Algorithm~\ref{alg:Search}.

\begin{algorithm}
 \caption{Search}
 \label{alg:Search}
 \KwIn{$(K_{\Sigma},\mathrm{Q},\sigma;\mathrm{EDB})$}
 \KwOut{$ \mathrm{R} $}
 \textbf{DU:} \\
  $H(\mathcal{P}(\mathcal{H}(R)))\rightarrow e_q$,
 $e_{q}\rightarrow (F_{Q},F_{R})$\;
 ${\rm QF}\{H,(F_{Q},F_{R})\}\rightarrow {\rm {\rm QF_Q}}$\;
 SHVE.KeyGen$(msk,{\rm QF_Q})\rightarrow \rm{S_{SHVE}}({\rm QF_Q})$\;
Send $\rm{S_{SHVE}}({\rm QF_Q})$ to the CS\;
 \textbf{CS:} \\
   \If{${\rm is\_occupied}(T\_F_{Q}) == \textup{false}$} {
            $\Return\ \textup{false}$ \ $/\ast$ {Checks\ if\ the\ element\ at\ the\ quotient\ index\ is\ occupied. If not, it returns false immediately.} $\ast/$\\
        } 
 $ s = {\rm find\_run\_index}({\rm QF}, {{F_Q}}) $;\ $/\ast$ {Finds the starting index of the \texttt{run} associated with the quotient.} $\ast/$\\

\Do{$({\rm is\_continuation}({\rm get\_elem}({\rm QF}, s)))$}{$rem = {\rm get\_remainder}({\rm get\_elem}({\rm QF}, s))$\;
\If  {$rem == {{F_R}}$}
	{  
	$ \Return \ \textup{true}$\;
	}

	\Else{$rem > {{F_R}} $\;
	$ \Return \ \textup{false}$\;
	}
  $s = {\rm incr} ({\rm QF},s)$$/\ast$ {Increments the index to access the next element in the \texttt{run}.}$\ast/$}	 
 $ \Return \ \textup{false}$\;
 
\end{algorithm}

\par \textbf{Update$(K_{\Sigma},O,up,\sigma;\mathrm{EDB})\rightarrow (K_{\Sigma},\sigma^{\prime};\mathrm{EDB}^{\prime})$}: Given $(K_{\Sigma},O,up,\sigma;\mathrm{EDB})$, the CS adds or deletes the object $O$ from the $\mathrm{EDB}$. For both addition and deletion, DO encodes the updated data and uses a hash function to obtain the updated fingerprint $e_a $ or $
e_d$ (for addition or deletion). DO then computes the quotient and remainder of the fingerprint. Do encrypt the quotient and remainder and sends them to the server for preparation to be inserted into the quotient filter. 
\begin{algorithm}
\caption{Update\_Addition}
\label{alg:Addition}
\KwIn{$(K_{\Sigma},O,add,\sigma;\mathrm{EDB})$}
\KwOut{$(K_{\Sigma},\sigma^{\prime};\mathrm{EDB}^{\prime})$}
 \textbf{DO:} \\
  $H({P}(\mathcal{H}(p_{a})))\rightarrow e_{a}$,
$e_{a}\rightarrow ({{F_{Q}}},{{F_{R}}})$\;
 ${\rm QF}\{H,(F_{Q},F_{R})\}\rightarrow {\rm QF_A}$\;
 SHVE.KeyGen$(msk,{\rm QF_A})\rightarrow \rm{S_{SHVE}}({\rm QF_A})$\;
Send $\rm{S_{SHVE}}({\rm QF_A})$ to the CS\;
 \textbf{CS:} \\
${\rm prev} = {\rm get\_elem}({\rm QF}, s)$ \;
${\rm empty} = {\rm is\_empty\_element}({\rm prev}) $\;

\Do{$(!{\rm empty})$}{
    \If{$(!{\rm empty})$}{
        ${\rm prev} = {\rm set\_shifted}({\rm prev})$ \;
        \If{$({\rm is\_occupied\_element}({\rm prev}))$}{
            ${\rm curr} = {\rm set\_occupied\_element}({\rm curr})$ \;
            ${\rm prev} = {\rm clr\_occupied\_element}({\rm prev})$ \;
        }
    }
    ${\rm set\_element}({\rm QF}, s, {\rm curr}) $\;
    ${\rm curr} ={\rm prev}$ ,
    $s = {\rm incr}({\rm QF}, s) $,
    ${\rm prev} = {\rm get\_elem}({\rm QF}, s)$ \;
    ${\rm empty} = {\rm is\_empty\_element}({\rm prev}) $\;
}
\If{${\rm(QF_{i}\_entries \ge \frac{1}{20}QF_{i}\_max\_size)}$}{
    $\mathrm{i} ++$;\ $/\ast$ {Expand the original ${\rm QF_{i}}$ to ${\rm QF_{i+1}}$ when ${\rm QF_{i}}$ exceeds 5\% of the capacity.} $\ast/$\\
}

$T\_{{F_{Q}}}$ = {\rm get\_elem}(QF, ${{F_{Q}}}$) \;
entry = (${{F_{R}}}$ $\ll$ 3) \& $ \sim $7 \;

\If{$({\rm is\_empty\_element}({T\_{F_{Q}}}))$}{
    set\_elem(QF, ${{F_{Q}}}$, ${\rm set\_occupied}(entry))$ \;
     $\mathrm{QF\_entries} ++$\;
    \Return true
}

\If{$(!{\rm is\_occupied}({T\_{F_{Q}}})) $}{
    set\_elem(QF, ${{F_{Q}}}$, set\_occupied(${T\_{F_{Q}}}$)) \;
}

${\rm start} = {\rm find\_run\_index}({\rm QF}, {{F_{Q}}})$ \;
$s = {\rm start} $\;

\If{$s != {{F_{Q}}}$}{
    ${\rm entry} = {\rm set\_shifted}({\rm entry})$ \;
}

${\rm insert\_into}({\rm QF}, s, {\rm entry})$ \;
 $\mathrm{QF\_entries} ++$\;
\Return true
\end{algorithm}

\par As for addition, let us define \texttt{prev} as the element currently at index $s$, and \texttt{curr} as the element that will be written to index $s$. If the slot corresponding to the quotient is empty, we add it directly and end. We set the \texttt{is\_occupied} metadata bit and find the starting position of the \texttt{run} corresponding to ${{F_{Q}}}$. Note that the \texttt{is\_occupied} metadata bit in ${{F_{Q}}}$ must be marked. If the slot is not empty, the CS executes the loop: set the previous slot as \texttt{is\_shifted}, if the previous slot is set as \texttt{is\_occupied}, set the current slot \texttt{curr} as \texttt{is\_occupied}, and clear the metadata \texttt{is\_occupied} of the previous slot \texttt{prev}. If the \texttt{is\_occupied} in ${T}_{{F_{Q}}}$ is not marked, then we return the expected starting position of the \texttt{run} (because the \texttt{run} corresponding to ${{F_{Q}}}$ does not exist). Otherwise, we return the starting position of the \texttt{run} (because the \texttt{run} corresponding to ${{F_{Q}}}$ exists). If the \texttt{run} corresponding to ${{F_{Q}}}$ exists, we must determine the specific insertion position to maintain the order of \texttt{run} after insertion. The result is in the variable $s$. If $s$ is the starting position of the \texttt{run}, we must set \texttt{is\_continuation} at the starting position since, after adding the fingerprint at $s$, this element becomes part of the continuation of the \texttt{run}; otherwise, we must set \texttt{is\_continuation} in the element being added. We determine whether the insertion position is a canonical slot. If not, then we set \texttt{is\_shifted} in \texttt{insert\_into} and then move elements one by one.

When the number of entries in ${\rm QF_{i}}$ exceeds 5\% of its maximum size, a new ${\rm QF_{i+1}}$ is created (usually twice the size), the value of quotient $q$ increases by 1 (due to capacity doubling), and the value of remainder $r$ remains unchanged.
The total number of bits in the new filter is $p = q_{new} + r$. For example, the quotient and remainder of fingerprint 110$\big|$0101 are 110 and 0101. In this old filter, $p = 7, r = 4, q = p - r = 3$. As for the new filter that doubles the size, the fingerprint and remainder stay unchanged. The quotient $q_{new} = q+1=3 + 1 = 4$. So the new quotient is taken as the highest four digits 1100, and the remainder is still 0101. A formal representation of the \textbf{Addition} phase is displayed in Algorithm~\ref{alg:Addition}.

\begin{algorithm}
\caption{Update\_Deletion}
\label{alg:Deletion}
\KwIn{$(K_{\Sigma},O,del,\sigma;\mathrm{EDB})$}
\KwOut{$(K_{\Sigma},\sigma^{\prime};\mathrm{EDB}^{\prime})$}
 \textbf{DO:} \;
  $H({P}(\mathcal{H}(p_{d}))) \rightarrow e_{d}$,
  $e_{d}\rightarrow ({{F_{Q}}},{{F_{R}}})$\;
  ${\rm QF} \{H,(F_{Q},F_{R})\} \rightarrow {\rm QF_D}$\;
  SHVE.KeyGen$(msk,{\rm QF_D}) \rightarrow \rm{S_{SHVE}}({\rm QF_D})$\;
  Send $\rm{S_{SHVE}}({\rm QF_D})$ to the CS\;

 \textbf{CS:} \\
 ${\rm get\_elem}({\rm QF}, s) \rightarrow {\rm curr}$,
 ${\rm incr}({\rm QF}, s) \rightarrow sp$,
 $s \rightarrow {\rm orig}$\;

 \While{$(\textup{true})$}{
    ${\rm is\_occupied}({\rm curr}) \rightarrow {\rm curr\_occupied}$\;

    \If{$({\rm is\_empty\_element}({\rm next}) \parallel {\rm is\_cluster\_start}({\rm next}) \parallel sp == {\rm orig})$}{
        ${\rm set\_elem}({\rm QF}, s, 0)$\;
        \textbf{return}\;
    } 
    \Else {
        \If{$({\rm is\_run\_start}({\rm next}))$}{
            \Do{$(!{\rm is\_occupied}({\rm {\rm get\_elem}({\rm QF}, {\rm F_{Q}}))})$}{
                ${\rm incr}({\rm QF}, {\rm F_{Q}}) \rightarrow {\rm F_{Q}}$\;
            }
            \If{$({\rm curr\_occupied} \&\& \ {\rm F_{Q}} == s)$}{
                ${\rm clr\_shifted}({\rm next}) \rightarrow {\rm updated\_next}$\;
            }
        }
        ${\rm set\_elem}({\rm QF}, s, {\rm curr\_occupied} ? {\rm set\_occupied}({\rm updated\_next}) : {\rm clr\_occupied}({\rm updated\_next}))$\;
        $sp \rightarrow s$\;
        ${\rm incr}({\rm QF}, sp) \rightarrow sp$\;
        ${\rm next} \rightarrow {\rm curr}$\;
    }
}
   
\If{$(!{\rm is\_occupied}(T_{F_Q})) || (!{\rm QF} \rightarrow {\rm QF_{entries}})$}{
     \Return true\;
}

${\rm start} = {\rm find\_run\_index}({\rm QF}, F_{Q})$\;
$s = {\rm start}$\;

\Do{$({\rm is\_continuation}({\rm get\_elem}({\rm QF}, s)))$}{
    $rem = {\rm get\_remainder}({\rm get\_elem}({\rm QF}, s))$\;
    
    \If{$rem == F_{R}$}{
        \textbf{break}\;
    }
    \ElseIf{$rem > F_{R}$}{
         \Return true\;
    }
    
    $s = {\rm incr}({\rm QF}, s)$\;
}

\If{$(rem != F_{R})$}{
    \Return true\;
}

delete\_entry$({\rm QF}, s, F_{Q})$\;

 $\mathrm{QF\_entries} --$\;
\Return true

\end{algorithm}
\par As for deletion, $sp$ represents the index of the entry after $s$, \texttt{curr} represents the value of the slot corresponding to $s$, \texttt{next} represents the value of the slot corresponding to $sp$, and \texttt{next} represents the original value of $s$. The CS retrieves the current element from position $s$ in the QF and obtains the next increment position of $s$, saving the original starting position. Then it enters an infinite loop to process entry sliding. First, it checks if the current slot \texttt{curr} is occupied. If the next slot is empty or marks the start of a cluster, the current slot \texttt{curr} is set as empty. Otherwise, the CS prepares to update the next slot \texttt{next}. If the next slot \texttt{next} marks the start of a run, the CS finds the next occupied quotient $F_{Q}$ and increments the quotient's position. If the current slot \texttt{curr} is occupied and the quotient equals the current position, the \texttt{is\_shifted} flag is cleared.

The CS keeps the slot \texttt{is\_occupied} and updates the position pointer 
$sp$, incrementing both the pointer 
$sp$ and the current slot \texttt{curr}. If the start of a run for the quotient is found at position $s$, the CS traverses consecutive slots to retrieve the remainder of the current slot \texttt{curr}. If a matching remainder is $rem$ found, the traversal terminates. If the remainder $rem$ is greater than the target remainder $F_{R}$, the CS returns true, indicating that no match was found. The CS then moves $s$ to the next slot and continues until the target remainder is found. Once the target remainder is found, the CS deletes the entry. A formal representation of the \textbf{Deletion} phase is displayed in Algorithm~\ref{alg:Deletion}.

\subsection{Trinity-\uppercase\expandafter{\romannumeral2}: Trinity with Improved Security and Accuracy} 

\subsubsection{Technique Overview}
The basic Trinity we proposed earlier is efficient enough, but as we said before, we want a secure and accurate and dynamic Trinity. So first, we use “salts" and CPRF techniques to avoid leakage from addition. To achieve this, Trinity employs a technique where hashed order tokens, denoted as $OT_i$, are used to introduce unique “salts" to the encrypted data points. This salting process is streamlined by utilizing a CPRF, which effectively minimizes bandwidth usage. Upon a new node being added to the EDB (specifically, the $(i+1)$th addition), a special ordering token, $OT_i$, is generated by the DO using the secret key $K$. This token is then used to “salts" the space-time node, enhancing its security. When a search query is initiated, the DUs provide both the ST and the query itself to the CS. The CS, in turn, leverages the ST to compute all necessary ordering tokens, subsequently reversing the salting process to unveil the original encrypted points.

We utilize the quotient filter data structure, a variant of the bloom filter. Like the bloom filter, the quotient filter also exhibits a false positive rate. And the FPR is an inherent challenge in filters, stemming from the design philosophy of the filter. It consists of a fixed-size binary bit array and a series of random mapping functions (hash functions). The core idea is to use multiple different hash functions to address conflicts. Due to the issue of hash collisions (where two different elements may map to the same value after applying a hash function), it introduces multiple hash functions to reduce collisions. If any hash function determines that an element is not in the set, then the element is definitely absent. The element likely exists only when all hash functions unanimously indicate its presence.

\begin{equation}
		\nonumber
		\textup{FPR}=\left ( 1-e^{-k\cdot \frac{n}{m} } \right ) ^{k},
	\end{equation}
 where $k$ is the number of hash functions, $n$ is the number of elements, $m$ is the length of the filter, $e$ is natural constant.
 \par Since false positives are inevitable in filter structures, we must implement a verification method to validate results. For ensuring the accuracy of results, we maintain a dedicated verify token for each file, linked to its unique ind. Bitmaps, known for their simplicity and efficiency, employ a binary bit array to represent information. The most prevalent bitmap indexing technique involves associating each bit with a specific element's position: a 1 signifies the element's presence within a set, while a 0 indicates its absence. Then we compress verify token with Roaring Bitmaps \cite{Chambi2014BetterBP} and encrypt it. Upon executing a search request, the verify tokens are presented to the DU alongside the other results. Having obtained these tokens, the DU decrypts and authenticates the results using decrypted verify token. The verification process not only reduces FPR but also decreases the quotient filter size, resulting in surprisingly significant space cost savings.
\subsubsection{Details of Trinity-\uppercase\expandafter{\romannumeral2} Construction} 
In this section, we will briefly introduce \textbf{Setup} and \textbf{Search}, focusing on \textbf{Addition} and \textbf{Verification}. \textbf{Deletion} phase in Trinity-\uppercase\expandafter{\romannumeral1} and Trinity-\uppercase\expandafter{\romannumeral2} is the same, so it will not be repeated here

\par \textbf{Setup$(\{0,1\}^{\lambda },0\rightarrow c)\rightarrow(K_{\Sigma},\sigma;\mathrm{EDB},OT_c)$:}
Given a security parameter $\lambda$, a update counter $c$ and a database DB, it returns a master key $msk$, $t$ hash functions $H=\{{H}_{i}(\cdot)\}_{i=1}^{t}$, an order token $OT_c$ and an encrypted database EDB. One of the key difference between Trinity-\uppercase\expandafter{\romannumeral1} and Trinity-\uppercase\expandafter{\romannumeral2} is the introduction of counter c, which tracks the number of updates. And the DO generates $OT_i$ by $\mathbb{G}$ with the input of $msk$ and $c$. Within the salting process, $OT_c$ is itself “salted" using a hash function H. Subsequently, the DU incorporates this salt into the original fingerprints, $e_{i}$, resulting in the generation of salted fingerprints, $e_{i}^c$.
And rest are the same to the Trinity-\uppercase\expandafter{\romannumeral1}. A formal representation of the \textbf{Setup} phase is displayed in Algorithm~\ref{alg:setup}.
\begin{algorithm}
\caption{Setup ($1^\lambda$)}
\label{alg:setup}
\KwIn{$(\{0,1\}^{\lambda }),c,\mathrm{DB})$}
\KwOut{$(K_{\Sigma},\sigma;OT_c,\mathrm{EDB})$}
$\mathrm{Setup(1^{\lambda})}\rightarrow msk$
\;
Randomly generates $H=\{{H}_{i}(\cdot)\}_{i=1}^{t}$\;
$K_{\Sigma}=\{msk,H\}$,
$\mathbb{G}(msk,c)\rightarrow OT_{c}$,
$H(K, OT_{c})\rightarrow e^{c}$\;
 \For {{\rm space-time node} $O_i=\{p_i,ind_i\}$ }{
 $H({P}(\mathcal{H}(p_{i})))\rightarrow e_{i}$,
$e_{i}\oplus e^{c}\rightarrow e_{i}^{c}$,
 $e_{i}^{c} \rightarrow ({F_{Q}^{c}},{F_{R}^{c}})$\;
 ${\rm QF}\{H,{{F_{Q}^{c}}},{{F_{R}^{c}}})\}\rightarrow {\rm QF_i}$\;
 SHVE.Enc$(msk,\textbf{“True"},{\rm QF_i})\rightarrow \rm{C_{SHVE}}({\rm QF_i})$\;
 $\{\rm{C_{SHVE}}({\rm QF_{i}}),ind_i\}_{i=1}^{N}\rightarrow $EDB\;

}

  \textbf{return} $(K_{\Sigma},\sigma;\mathrm{EDB},OT_c)$ \;
\end{algorithm}

\begin{algorithm}
 \caption{Search}
 \label{alg:search}
 \KwIn{$(K_{\Sigma},\mathrm{Q},\sigma;\mathrm{EDB})$}
 \KwOut{$ \mathrm{R} $}
 \textbf{DU:} \\
$\bar{\mathbb{G}}(msk,c)\rightarrow ST$\;

Send $(ST,K,\rm{S_{SHVE}}({\rm QF_Q}))$ to the CS\;
 \textbf{CS:} \\
 $ \mathbb{G}(ST,c)\rightarrow OT_c$,
 $H(K, OT_{c})\rightarrow e^{c}$,
  $H(\mathcal{P}(\mathcal{H}(R)))\rightarrow e_q$\;
$e_{q}\oplus e^{c}\rightarrow e_{q}^{c}$,
  $e_{q}^{c} \rightarrow ({F_{Q}^{c}},{F_{R}^{c}})$,
 ${\rm QF}\{H,F_{Q}^{c},F_{R}^{c})\}\rightarrow {\rm QF_Q}$\;
 SHVE.KeyGen$(msk,{\rm QF_Q})\rightarrow \rm{S_{SHVE}}({\rm QF_Q})$\;
   \If{${\rm is\_occupied}(T\_F_{Q}) == \textup{false}$} {
            $\Return\ \textup{false}$ \ $/\ast$ {Check if the element at the quotient index is occupied. Return false if unoccupied.} $\ast/$\\
        } 
 $ s = {\rm find\_run\_index}({\rm QF}, {{F_Q}}) $;\ $/\ast$ {Finds the starting index of the \texttt{run} associated with the quotient.} $\ast/$\\

\Do{$({\rm is\_continuation}({\rm get\_elem}({\rm QF}, s)))$}{$rem = {\rm get\_remainder}({\rm get\_elem}({\rm QF}, s))$\;
\If  {$rem == {{F_R}}$}
	{  
	$ \Return \ \textup{true}$\;
	}

	\Else{$rem > {{F_R}} $\;
	$ \Return \ \textup{false}$\;
	}
  $s = {\rm incr} ({\rm QF},s)$$/\ast$ {Increments the index to move to the next element in the \texttt{run}.}$\ast/$}	 
 $ \Return \ \textup{false}$\;
 
\end{algorithm}
\par \textbf{Search $(K_{\Sigma},\mathrm{Q},\sigma;\mathrm{EDB})\rightarrow \mathrm{R}$:}
Trinity-\uppercase\expandafter{\romannumeral1} and Trinity-\uppercase\expandafter{\romannumeral2} are highly similar in the search phase, with the only difference being that Trinity-\uppercase\expandafter{\romannumeral2} employs $K$ and $e_{q}^{c}$ for salting desalting. Search token $e_{q}^{c}$ are like $e_{i}^{c}$ but input with a range query instead of space-time node. After the search, all added nodes in QF.Cache are desalted and transferred to ${\rm QF_i}$. A formal representation of the \textbf{Search} phase is displayed in Algorithm~\ref{alg:search}.
\begin{algorithm}
\caption{Addition}
\label{alg:addition}
\KwIn{$(K_{\Sigma},O,add,\sigma;\mathrm{EDB})$}
\KwOut{$(K_{\Sigma},\sigma^{\prime};\mathrm{EDB}^{\prime})$}
 \textbf{DO:} \\
$\mathbb{G}(msk,c)\rightarrow OT_{c}$,

$H(K, OT_{c})\rightarrow e^{c}$,
  $H({P}(\mathcal{H}(p_{a})))\rightarrow e_{a}$\;
 $e_{a}\oplus e^{c}\rightarrow e_{a}^{c}$,
 $e_{a}^{c} \rightarrow ({F_{Q}^{c}},{F_{R}^{c}})$\;
Send $(c,F_{Q}^{c},F_{R}^{c})$ to the CS\;
 \textbf{CS:} \\
 $(c,F_{Q}^{c},F_{R}^{c})\rightarrow $QF.Cache\;

${\rm prev} = {\rm get\_elem}({\rm QF_{i}}, s)$ \;
${\rm empty} = {\rm is\_empty\_element}({\rm prev}) $\;

\Do{$(!{\rm empty})$}{
    \If{$(!{\rm empty})$}{
        ${\rm prev} = {\rm set\_shifted}({\rm prev})$ \;
        \If{$({\rm is\_occupied\_element}({\rm prev}))$}{
            ${\rm curr} = {\rm set\_occupied\_element}({\rm curr})$ \;
            ${\rm prev} = {\rm clr\_occupied\_element}({\rm prev})$ \;
        }
    }
    ${\rm set\_element}({\rm QF}, s, {\rm curr}) $\;
    ${\rm curr} ={\rm prev}$ ,
    $s = {\rm incr}({\rm QF}, s) $,
    ${\rm prev} = {\rm get\_elem}({\rm QF}, s)$ \;
    ${\rm empty} = {\rm is\_empty\_element}({\rm prev}) $\;
}
\If{${\rm(QF_{i}\_entries \ge \frac{1}{5}QF_{i}\_max\_size)}$}{
    $\mathrm{i} ++$;\ $/\ast$ {Expand the original ${\rm QF_{i}}$ to ${\rm QF_{i+1}}$ when ${\rm QF_{i}}$ exceeds half capacity.} $\ast/$\\
}
$T\_{{F_{Q}}}$ = get\_elem(${\rm QF_{i}}, {{F_{Q}^{c}}}$) \;
entry = (${{F_{R}^{c}}}$ $\ll$ 3) \& $ \sim $7 \;

\If{$({\rm is\_empty\_element}({T\_{F_{Q}}}))$}{
    set\_elem$({\rm QF_{i}}, {{F_{Q}^{c}}}, {\rm set\_occupied}({\rm entry}))$ \;
     $\mathrm{QF\_entries} ++$\;
    \Return true
}

\If{$(!{\rm is\_occupied}({T\_{F_{Q}}})) $}{
    ${\rm set\_elem}({\rm QF_{i}}, {{F_{Q}^{c}}}, {\rm set\_occupied}({T\_{F_{Q}}}))$ \;
}

${\rm start} = {\rm find\_run\_index}({\rm QF_{i}}, {{F_{Q}^{c}}})$ \;
$s = {\rm start} $\;
\If{$s != {{F_{Q}^{c}}}$}{
    ${\rm entry} = {\rm set\_shifted}({\rm entry})$ \;
}

${\rm insert\_into}({\rm QF_{i}}, s, {\rm entry})$ \;
$\mathrm{QF\_entries} ++$\;
\textbf{return} true\;
\end{algorithm}

\par \textbf{Update $(K_{\Sigma},O,add,\sigma;\mathrm{EDB})\rightarrow (K_{\Sigma},\sigma^{\prime};\mathrm{EDB}^{\prime})$:} Since delete operation in Trinity-\uppercase\expandafter{\romannumeral1} and Trinity-\uppercase\expandafter{\romannumeral2} are exactly the same, so here we discuss add operation only. Every time the DO wants to add a space-time node $p=(x,y,z)$ to EDB, first, the DO increases the value of the counter to the match node O. Then the DO starts the “salted" process like \textbf{Setup} phase. Next, the DO sends the salted token into the QF.Cache of the server like in Fig.~\ref{quotient SBF}. If the number of entries in QF exceeds 20\% of its maximum size, expand the QF twice as the original size. The salted values $(e_{a}^{c},F_{Q}^{c},F_{R}^{c})$ are desalted and sent to ${\rm QF_{i}}$ after the next search. 

 The remaining operations are identical to those in Trinity-\uppercase\expandafter{\romannumeral1}. A formal representation of the \textbf{Addition} phase is displayed in Algorithm~\ref{alg:addition}.
\begin{figure}[htbp]
	
	\centering
	\includegraphics[scale=0.20]{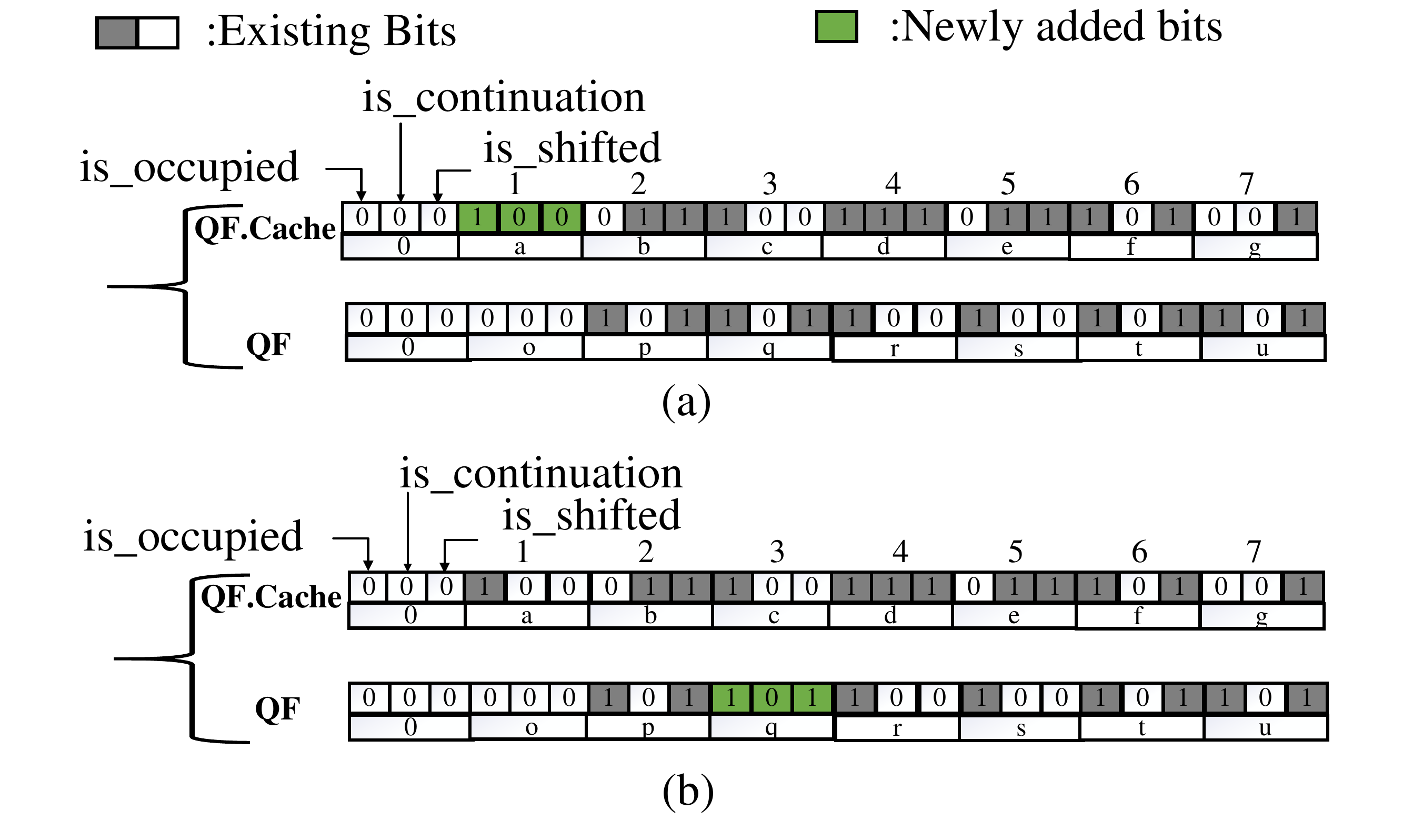}
	\caption{An example of addition}
	\label{quotient SBF}
\end{figure}
\par \textbf{Verification:}
  As mentioned earlier in technique overview, in order to maintain a low false positive rate of about 0.01$\%$, it is customary to set the ratio $N/m$ at 20, and the parameter $t$ at 14. This implies that for a database with a capacity of 5 million space-time nodes, a filter size of 100 million would be required. Such a configuration results in a considerable waste of storage space.But with the help of verification, we do not require such low FPR, we only have to keep a basic accuracy of quotient filter that doesn't affect normal search function. Based on verification, we set FPR$\leq 6\% $, and $t=-ln(p)/ln(2)$, we have 4.05. And assume we have 1 millions elements in the filter, we should set the length of filter as $N=mt/ln(2)$, the number is 5.85 millions. So we set $N/m=6$, $t=4$, and the FPR is 5.56$\%$. Consequently, both storage and computational costs decrease precipitously. Next, I will formally describe the Trinity-\uppercase\expandafter{\romannumeral2}. Verification Process, which consists of two distinct phases: Setup and Search.
\begin{itemize}
\item \textbf{Setup}: An additional verify array is generated for each space-time node ${V}(\mathcal{H}(p_i))\rightarrow va_i$. An extra secret key $sk$ is generated to encrypt verify arrays. Then we compress the verify array $\textbf{comp}(vt_i)\rightarrow vt_{i}^{\prime}$ and encrypt it $\textbf{Enc}(va_{i}^{\prime})\rightarrow vt_{i}$. Finally, we store $vt_{i}$ associated with the $ind_i$. The rest procedure are the same to the Trinity-\uppercase\expandafter{\romannumeral2}. 
\item \textbf{Search}: The CS sends verify token $vt_i$ along with other results to the DU. Once receiving results, the DU decrypts the $vt_i$ and verifies the token according to the query, if it matches, keep the results; if not, remove the index.
\end{itemize}
\section{Security Analysis}
This section delves into security analysis of our scheme. We begin by defining a leakage function, and then proposing a rigorous proof of the construction's security.

\par  The leakage function, denoted as $\mathcal{L}$, encompasses leakage from three distinct phases: setup, search, and update operations.
$\mathcal {L}=\lbrace \mathcal{L}^{Stp},\mathcal {L}^{Srch},\mathcal{L}^{Updt}\rbrace$. Meanwhile, the leakage function can also be explained as follows:
\par Access pattern: The access pattern, denoted as ${ap}$, signifies the specific locations of documents that align with the query.
\par Size pattern: The size pattern, denoted as $\mathbb{S}$, reflects in the number of objects $m$.
\begin{theorem}\label{leakage}
	
	Define the leakage function $\mathcal{L}_1$ of Trinity-{\uppercase\expandafter{\romannumeral1}}, if it can be described as following: 
	\begin{equation}
		\nonumber
		\mathcal{L}_{\uppercase\expandafter{\romannumeral1}}= \{\mathbb{S}(DB),ap\},
	\end{equation}
 
	Trinity-{\uppercase\expandafter{\romannumeral1}} is Indistinguishability under Selective Chosen-Plaintext Attack (IND-SCPA) secure if SHVE is IND-SCPA secure.
	
\end{theorem}
\begin{table*}[htbp]\footnotesize
	\caption{Complexity comparisons}
	\centering
	\label{Comparison}
 \resizebox{\linewidth}{!}{
	\begin{threeparttable}
		{	\begin{tabular}{ | c | c | c | c | c | c |}
				\hline	
				
				\multicolumn{1}{|c|}{ \multirow{2}*{$Scheme$} }& \multicolumn{2}{c|}{Computation} & \multicolumn{2}{c|}{Communication}& \multirow{2}*{Storage size}\\
				\cline{2-5} 
				$ $&$Search$&$Update$&$Search$&$Update$&\\
               
                \hline
				GRS-\uppercase\expandafter{\romannumeral 2}& $\mathcal {O}(\log_{ }{R}\cdot N)$& $\mathcal {O}(2^t N)$&$\mathcal {O}(\log_{ }{R}\cdot N)$& $\mathcal {O}(2^t N)$&$\mathcal {O}(2^t N)$ \\
               \hline
				$\mathsf {DSSE}_{\mathsf {SKQ}}$ &$\mathcal {O}(a_{w}(|R|+k))$& $\mathcal {O}(\log N+\log m)$&$\mathcal {O}((1+k)l)$& $\mathcal {O}((\log N+\log m)(2\lambda +l))$&$\mathcal {O}(a_{w}(\log N+\log m)(2\lambda +l))$ \\
                \hline
				SKSE-{\uppercase\expandafter{\romannumeral2}} &$\mathcal {O}(\overline{x}k m\cdot \log m)$& $N/A$&$  \mathcal {O}(\overline{x}k\lambda m\cdot \log m )$& $N/A$&$\mathcal {O}(m\lambda N)$ \\
                \hline
                Trinity-{\uppercase\expandafter{\romannumeral1}} &$\mathcal {O}(k m\cdot \log m)$& \textcolor{red}{$\mathcal {O}(k)$}&$  \mathcal {O}(k\lambda m\cdot \log m )$& $\mathcal {O}(k \lambda)$&$\mathcal {O}(m\lambda N)$ \\
                \hline
                Trinity-{\uppercase\expandafter{\romannumeral2}} &$\mathcal {O}(k (m\cdot \log m+\log c))$& \textcolor{red}{$\mathcal {O}(k)$}&$ \mathcal {O}(k\lambda (m\cdot \log m+\log c))$& $\mathcal {O}(k \lambda)$&$\mathcal {O}(m\lambda N)$ \\
				\hline

				
		\end{tabular}}
		
		\begin{tablenotes}
			\footnotesize
			\item[ ] \textbf{Notes.} $N$ is the number of data points in the database(\textit{i.e.} the number of entries), $M$ is the number of keywords in the database, $m$ is the size of  database, $|U|$ is the size of each ciphertext,  $\lambda$ is a security parameter, $k$ is the number of returned results range, $|k|$ is the size of key $K$, $|\varepsilon|$ is the size of random vector $\varepsilon$, $\overline{x}$ is the average number of the nodes at each level that traversed by the CS, $c$ is the node number of GGM tree and $R$ is search range radius. $2^t$ is the size of the binary tree.
		\end{tablenotes}
	\end{threeparttable}
 }
\end{table*}

\par \begin{proof} The proof of Theorem \ref{leakage} is established through a simulation-based approach. Assuming that the adversary $\mathcal{A}$ cannot differentiate between the output of the real game and the ideal game, we can conclude that no leakage exists beyond $\mathcal{L}_1$. All spatio-temporal data are encrypted by SHVE which is proved in \cite{li2021secure}.
The security of Trinity-{\uppercase\expandafter{\romannumeral1}} is based on the security of SHVE. So the proof of Theorem \ref{leakage} leverages a simulation as follows:
\begin{itemize}
	\item \textbf{Setup}: The adversary $\mathcal{A}$ transmits a chosen database, denoted as ${\bf DB}=\{{\bf p}_{i},\ {\bf ind}_{i}\}_{i=1}^{N}$, to the challenger $\mathcal{C}$. The challenger $\mathcal{C}$ generates both a set of $t$ random hash functions H and a secret key $msk$, maintaining the secrecy of $K_{\Sigma}=\{msk,H\}$. 
\item \textbf{Phase 1}: The adversary $\mathcal{A}$ proceeds to adaptively select a series of queries, each represented as $Q_j$, where $j \in [q_1]$. In response to each query, $\mathcal{C}$ encodes the range query $Q_j$ into ${P}(\mathcal{H}(Q_j)))$. Subsequently, for each prefix element $pe_k$ within ${P}(\mathcal{H}(Q_j)))$ (where $1 \leq k \leq \beta$), $\mathcal{C}$ constructs a quotient filter comprises $({{H}}({p}_{i}),{{F_{Q}}},{{F_{R}}})$, and then encrypts it with SHVE, with $'0'$ positions being substituted with $'\star'$ and $'1'$ positions remaining unaltered. Finally, $\mathcal{C}$ transmits the search token $ST_j$ to $\mathcal{A}$.
\item \textbf{Challenge}: $\mathcal{C}$ randomly selects a bit, denoted as $b$, from the set $\{0, 1\}$. For each space-time node $O_i$, $\mathcal{C}$ constructs a corresponding quotient filter, $QF\{H,(e_{i},{{F_{Q}}},{{F_{R}}})\}\rightarrow {\rm QF_i}$. Subsequently, $\mathcal{C}$ offers $\mathcal{A}$ with the challenge encrypted database EDB, constructed by SHVE.Enc$(msk,{\rm QF_i})\rightarrow \rm{C_{SHVE}}({\rm QF_i})$ for $i \in [N]$.
\item \textbf{Phase 2}: $\mathcal{A}$ repeats the \textbf{Phase 1} procedure and receives $ST_j$ for $q_1 +1\leq j\leq q_2$.
\item \textbf{Guess}: $\mathcal{A}$ takes a guess of $b$.
\end{itemize}

\par Due to the utilization of SHVE for encryption within our Trinity-{\uppercase\expandafter{\romannumeral1}}, the indistinguishability of indexes and trapdoors within the Trinity-{\uppercase\expandafter{\romannumeral1}} directly inherits the indistinguishability property of SHVE. Notably, the security game associated with our Trinity-{\uppercase\expandafter{\romannumeral1}} is essentially simulated by $q^{\prime}$ instances of SHVE, where $q^{\prime}$ denotes the total number of SHVE instances engaged in the game. As a result, the ability of adversary $\mathcal{A}^{\prime}$ to distinguish between two SHVE encrypted values would directly enable it to distinguish between indexes and trapdoors within the Trinity-{\uppercase\expandafter{\romannumeral1}}. This relationship can be succinctly expressed as follows:
\begin{equation}
	\nonumber
 \begin{array}{l}{{\mathbf{Adv}_{\mathrm{Trinity-{\uppercase\expandafter{\romannumeral1}}},\mathcal{A}}^{\mathrm{IND-SCPA}}(1^{\lambda})\leq\mathbf{Adv}_{\mathrm{SHVE},\mathcal{A}^{\prime}}^{\mathrm{IND-SCPA}}(1^{\lambda})}} {{\leq q^{\prime}\cdot\mathrm{negl}(\lambda)}}\end{array}.
\end{equation}
\par Since Trinity-{\uppercase\expandafter{\romannumeral2}} also utilizes SHVE as its encryption tool, its security can be proven using the same approach. 
\end{proof}
\par We formally define forward security for Trinity-{\uppercase\expandafter{\romannumeral2}}, where CPRF and “salts" work in tandem to achieve this property. This combination effectively thwarts adaptive file injection attacks, as rigorously proven by the following demonstration.

\par \textbf{Forward Security.} 
Forward security prevents information leakage during the addition of new data in dynamic updates. For clarity, we adopt the precise definition from \cite{bost2016ovarphiovarsigma}.
\begin{theorem}
	
If the function $\mathcal{L}^{FS}=\{\mathcal{L}^{Updt}\}$ is defined as follows, Trinity-{\uppercase\expandafter{\romannumeral2}} can be considered as forward-secure: 
	\begin{equation}
    \begin{split}
		\nonumber
   \mathcal{L}^{Updt}(op, in)= \mathcal{L}^{\prime}(op, {(ind_i, c)}). \\
    \end{split}
	\end{equation}

	The function $\mathcal{L}^{\prime}$ is stateless and $c$ represents the number of updated keywords for the updated file $ind_i$.
	
\end{theorem}
 \begin{proof}  
 First, spatio-temporal data is encrypted locally by the DO using SHVE before being transmitted to the CS. Then, Trinity-{\uppercase\expandafter{\romannumeral2}} uses SHVE to ensure document security while keeping the key private. In this way, the master key $msk$ of SHVE is held only by the DO, thus preventing unauthorized access by the CS. Indices undergo encryption locally before being uploaded to the CS for security and confidentiality. The quotient filter stores only fingerprints, quotients and remainders. And ST generation encrypts each range query using SHVE, CPRF and “salts", ensuring confidentiality. the DO incorporates “salts" alongside SHVE when adding nodes to the quotient filter, effectively preventing the CS from inferring keyword information from past search tokens.
\end{proof}

\section{Performance Evaluation}
We implemented Trinity in C++\footnote{https://github.com/eulermachine/trinity/}, with SHVE utilizing multi-threading technology, CPRF implemented via HMAC256, and Murmur serving as the hash function. In this section, a comprehensive comparison is made between the construction performance, volume size, token generation efficiency, search capability, addition performance, and deletion performance of Trinity-\uppercase\expandafter{\romannumeral 1} and Trinity-\uppercase\expandafter{\romannumeral 2}, as compared to those of SKSE-\uppercase\expandafter{\romannumeral 2} \cite{wang2021enabling}, GRS-\uppercase\expandafter{\romannumeral 2} \cite{kermanshahi2020geometric}, and $\mathsf {DSSE}_{\mathsf {SKQ}}$ \cite{wang2022forward}. We conducted experiments on an Ubuntu 20.04.5 LTS system using an Intel Xeon Gold 6226R processor. This processor has a base clock speed of 2.90 GHz. The system featured 64*8=512 GB of RAM. Network bandwidth used during these experiments was set at 100 Mbps.
\subsection{Theoretical Analysis} Table~\ref{Comparison} provides a comprehensive comparison of the time and space complexity of  GRS-\uppercase\expandafter{\romannumeral 2}, $\mathsf {DSSE}_{\mathsf {SKQ}}$, SKSE-\uppercase\expandafter{\romannumeral 2}, Trinity-\uppercase\expandafter{\romannumeral 1} and Trinity-\uppercase\expandafter{\romannumeral 2}. Since we use quotient filter as search structure and there are k queries sent by the DUs, the Trinity-\uppercase\expandafter{\romannumeral 1}'s computation complexity of search is $k m\cdot \log m$. So is computation complexity of update. As for communication complexity of search, it is $\mathcal {O}(k\lambda m\cdot \log m )$. The communication complexity of update is $\mathcal {O}(k\lambda m)$. Trinity-\uppercase\expandafter{\romannumeral 2} differs in two aspects: the use of “salts" and QF.Cache. The employed GGM-tree requires $\mathcal {O}(\log c )$ computation time to retrieve ST or OT. So we add $\mathcal {O}(\log c )$ and $\mathcal {O}(\lambda \log c )$ to the computation complexity and communication complexity of search in Trinity-\uppercase\expandafter{\romannumeral 2}. And since we set up a QF.Cache, there is an extra size of EDB, but we set the size of QF.Cache to $\mathcal {O}(\log m)$, which is negligible.

\subsection{Datasets} We use a spatio-temporal dataset from Yelp to evaluate the performance. We selected 1,007,016 timestamped locations, with some locations appearing multiple times at different timestamps. In the context of fine-grained spatiotemporal data, we refer to the logistics transportation scenario and set the spatial granularity to 10 meters and the temporal granularity to 5 minutes\cite{liang2019urbanfm}.
\begin{figure}[htbp]
\centering{
\subfloat[EDB setup latency]{\includegraphics[width=.48\columnwidth]{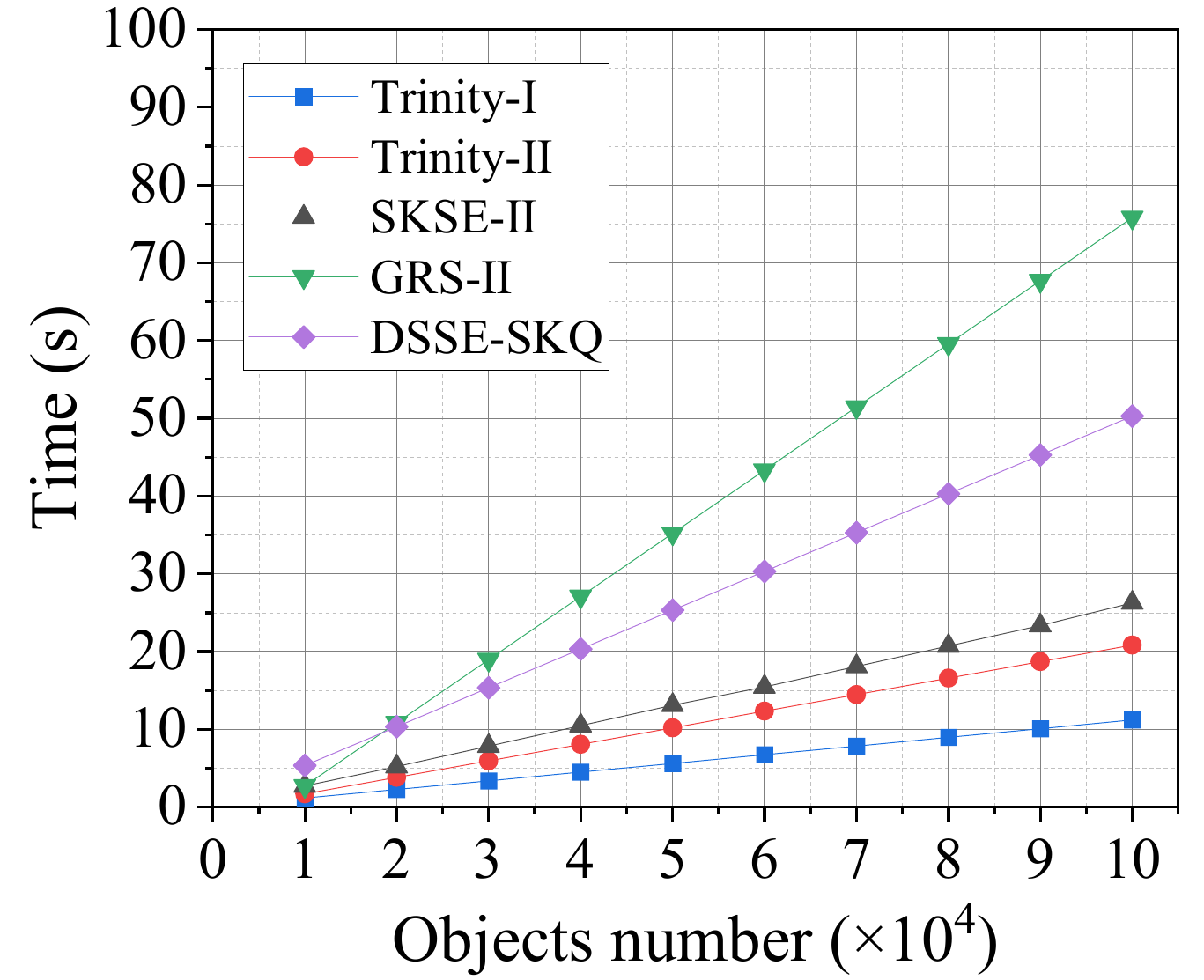}\label{fig:5_build} }
\subfloat[EDB size]{\includegraphics[width=.50\columnwidth]{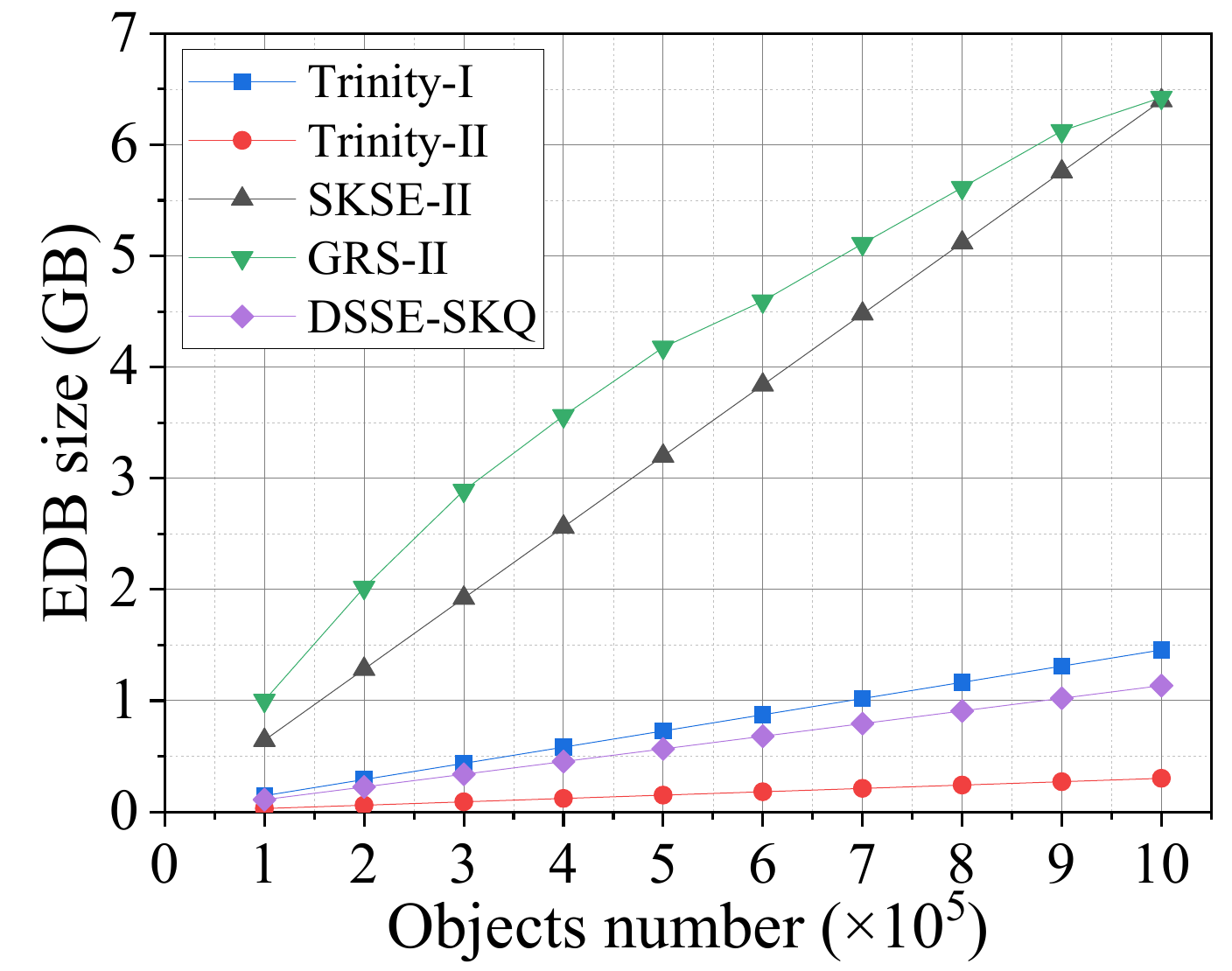}\label{fig:5_EDB} }

}
\caption{Trinity-\uppercase\expandafter{\romannumeral 1} vs Trinity-\uppercase\expandafter{\romannumeral 2},SKSE-\uppercase\expandafter{\romannumeral 2},GRS-\uppercase\expandafter{\romannumeral 2} and $\mathsf {DSSE}_{\mathsf {SKQ}}$ setup performance}
\label{fig:5_comp} 
\end{figure}

\subsection{Setup Performance}
\par \textbf{Setup latency.} The setup latency of Trinity-\uppercase\expandafter{\romannumeral 1} and Trinity-\uppercase\expandafter{\romannumeral 2} remains stable and increases linearly, as do the setup latency of GRS-\uppercase\expandafter{\romannumeral 2}, $\mathsf {DSSE}_{\mathsf {SKQ}}$, and SKSE-\uppercase\expandafter{\romannumeral 2}. For instance, as shown in Figure~\ref{fig:5_comp}\subref{fig:5_build}, Trinity-\uppercase\expandafter{\romannumeral 1} takes about 11 seconds to build an EDB for 100,000 data points, which is significantly less than the time taken by Trinity-\uppercase\expandafter{\romannumeral 2}, which takes 20 seconds. Meanwhile, SKSE-\uppercase\expandafter{\romannumeral 2}, GRS-\uppercase\expandafter{\romannumeral 2}, and $\mathsf {DSSE}_{\mathsf {SKQ}}$ currently take 26 seconds, 75 seconds, and 50 seconds respectively to construct an EDB of 100,000 data points. To ensure accuracy, each method was tested 100 times, with the average value serving as the final result.

\par \textbf{EDB size.}  The size of Trinity-\uppercase\expandafter{\romannumeral 1} and Trinity-\uppercase\expandafter{\romannumeral 2} expands proportionally to the number of data points, as does the size of SKSE-\uppercase\expandafter{\romannumeral 2}, GRS-\uppercase\expandafter{\romannumeral 2} and $\mathsf {DSSE}_{\mathsf {SKQ}}$. As illustrated in Figure~\ref{fig:5_comp}\subref{fig:5_EDB}, when handling 100,000 data points, Trinity-\uppercase\expandafter{\romannumeral 1} and $\mathsf {DSSE}_{\mathsf {SKQ}}$ occupy EDB sizes of 1.456 GB and 1.136 GB respectively, which are relatively similar and notably five times smaller than the EDB sizes consumed by GRS-\uppercase\expandafter{\romannumeral 2} (6.395 GB) and SKSE-\uppercase\expandafter{\romannumeral 2} (6.424 GB). However, Trinity-\uppercase\expandafter{\romannumeral 2} boasts a much more compact EDB size of 0.30146 GB due to its efficient roar bitmap architecture, which dramatically reduces storage cost.
\subsection{Search Performance}
\par \textbf{Token generation performance.} In evaluating the token generation cost of Trinity-\uppercase\expandafter{\romannumeral 1} and Trinity-\uppercase\expandafter{\romannumeral 2}, we compare their token generation times with those of Figure~\ref{fig:6_comp}\subref{fig:5_enc} for SKSE-\uppercase\expandafter{\romannumeral 2}, GRS-\uppercase\expandafter{\romannumeral 2}, and $\mathsf {DSSE}_{\mathsf {SKQ}}$. The evaluation tested different numbers of entries, ranging from 10,000 to 100,000 space-time nodes. Trinity-\uppercase\expandafter{\romannumeral 1} and Trinity-\uppercase\expandafter{\romannumeral 2} exhibit minimal differences in token generation cost, with respective times of 19.748 ms and 20.198 ms per token generation.
In terms of encryption methods, SKSE-\uppercase\expandafter{\romannumeral 2} employs HVE, while GRS-\uppercase\expandafter{\romannumeral 2} and $\mathsf {DSSE}_{\mathsf {SKQ}}$ utilize ASHE. Their corresponding token generation times are 92.024 ms, 3658.275 ms, and 2347.0956 ms, respectively. To ensure accuracy, each method was tested 100 times, with the average value serving as the final result.
\begin{figure}[htbp]
\centering{
\subfloat[Token generation latency]{\includegraphics[width=.52\columnwidth]{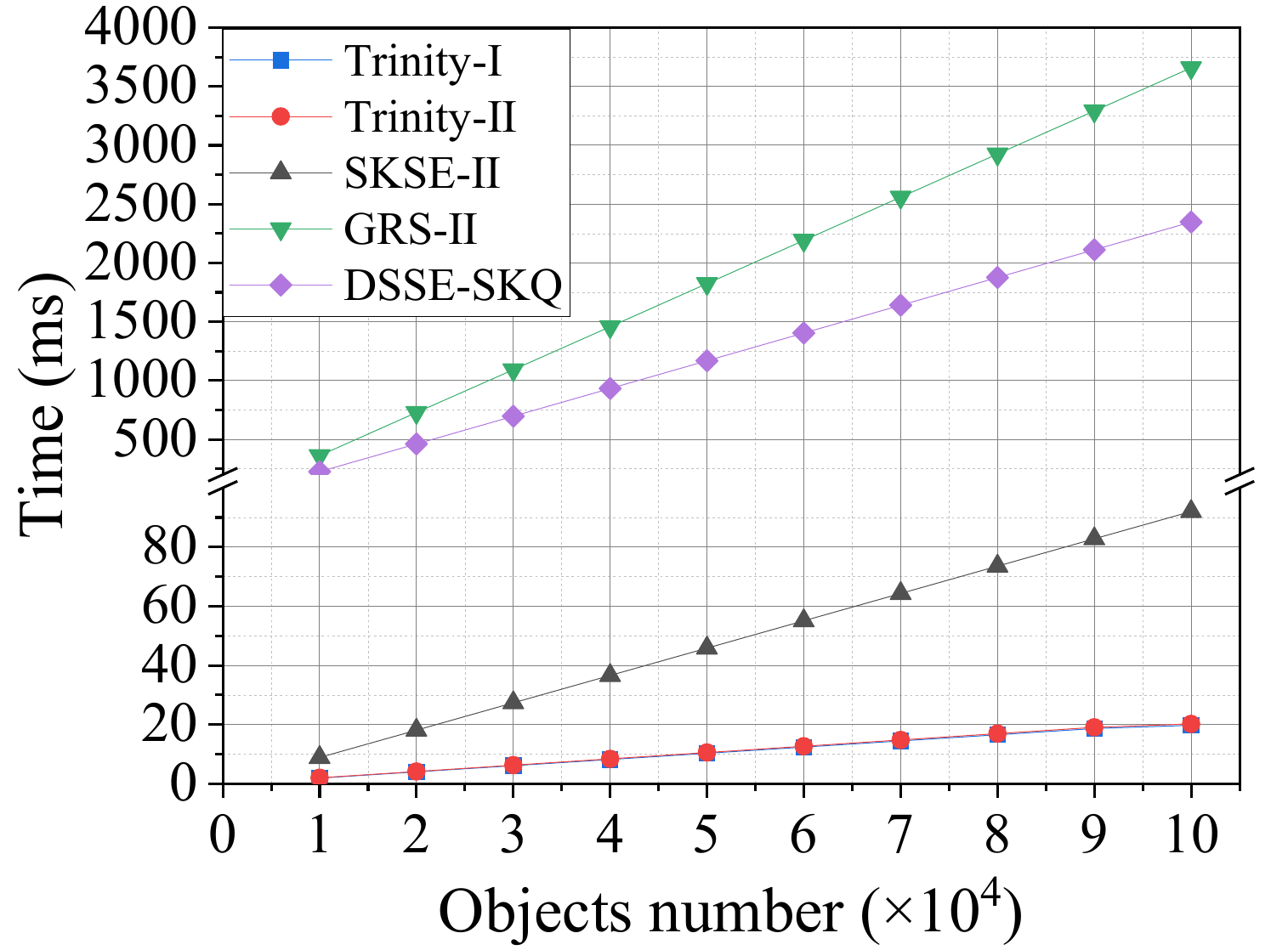}\label{fig:5_enc} }
\subfloat[Search latency]{\includegraphics[width=.47\columnwidth]{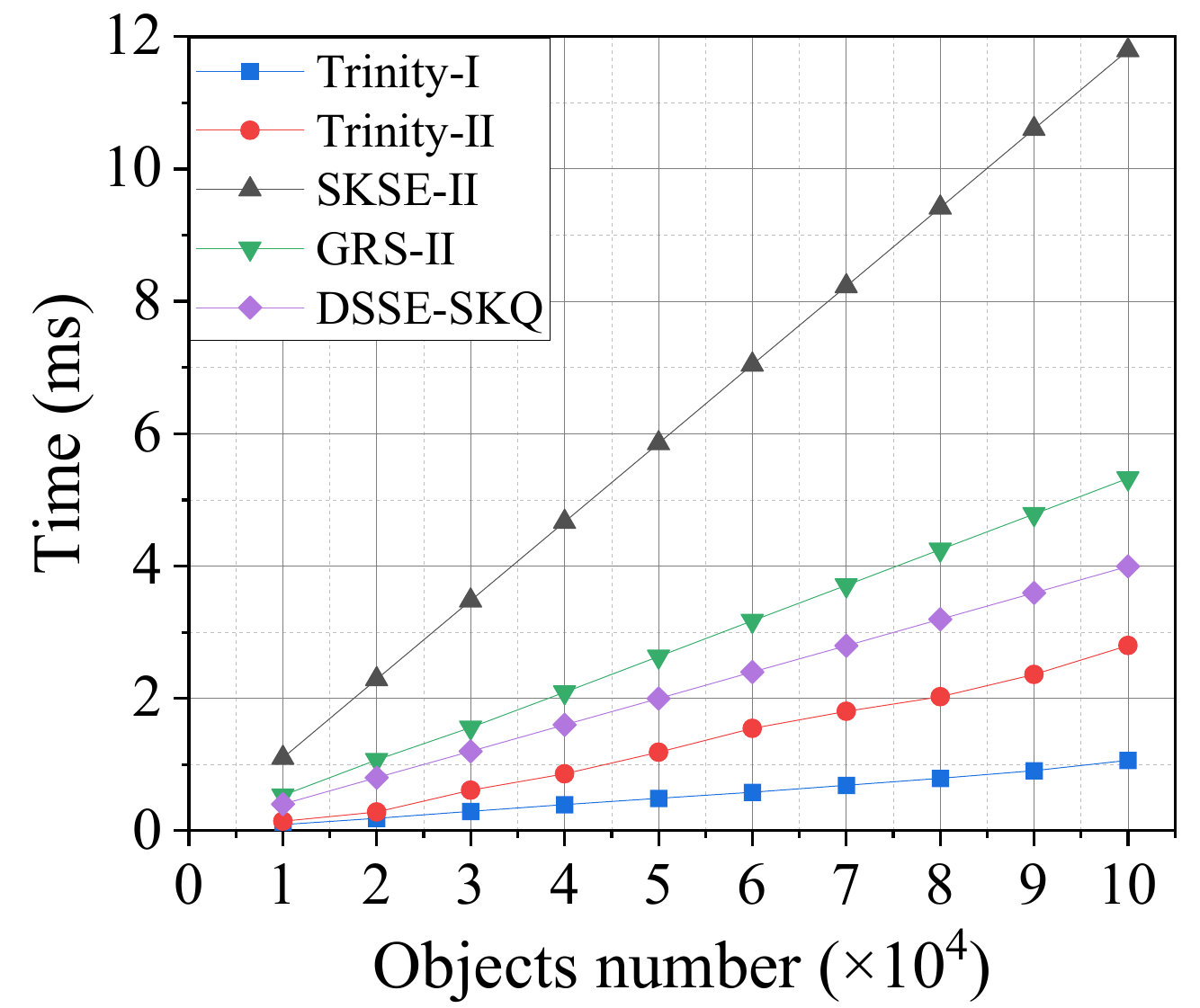}\label{fig:5_search} }
}
\caption{Trinity-\uppercase\expandafter{\romannumeral 1} vs Trinity-\uppercase\expandafter{\romannumeral 2},SKSE-\uppercase\expandafter{\romannumeral 2},GRS-\uppercase\expandafter{\romannumeral 2} and $\mathsf {DSSE}_{\mathsf {SKQ}}$ search performance}
\label{fig:6_comp} 
\end{figure}

\par \textbf{Search latency.} In evaluating the efficiency of Trinity-\uppercase\expandafter{\romannumeral 1} and Trinity-\uppercase\expandafter{\romannumeral 2}, we compare their search times to those of SKSE-\uppercase\expandafter{\romannumeral 2}, GRS-\uppercase\expandafter{\romannumeral 2}, and $\mathsf {DSSE}_{\mathsf {SKQ}}$, as depicted in Figure~\ref{fig:6_comp}\subref{fig:5_search}. The number of entries, which refers to space-time objects, ranges from 10,000 to 100,000 in increments of 10,000. The search range is set to 100 $m^2$ in 5 minutes. As the number of entries increases, the search time cost per entry in Trinity-\uppercase\expandafter{\romannumeral 2} increases once the number exceeds 10\% of its current capacity. This results in a change from a search time cost per entry of 14.262 ns when there are 10,000 entries to 27.994 ns when there are 100,000 entries.

In contrast, the average search time per entry for Trinity-\uppercase\expandafter{\romannumeral 1}, SKSE-\uppercase\expandafter{\romannumeral 2}, GRS-\uppercase\expandafter{\romannumeral 2}, and $\mathsf {DSSE}_{\mathsf {SKQ}}$ is respectively 10.62 ns, 110.1 ns, 53.162 ns, and 40.1 \textmu s. Consequently, it can be inferred that Trinity-\uppercase\expandafter{\romannumeral 1} outperforms Trinity-\uppercase\expandafter{\romannumeral 2}, being 2.63 times faster when there are 100,000 entries. Additionally, Trinity-\uppercase\expandafter{\romannumeral 2} is also 4.21 times faster than SKSE-\uppercase\expandafter{\romannumeral 2}, 1.9 times faster than GRS-\uppercase\expandafter{\romannumeral 2}, and 1.426 times faster than $\mathsf {DSSE}_{\mathsf {SKQ}}$.
\begin{figure}[htbp]
\centering{

\subfloat[Addition performance]{\includegraphics[width=0.50\columnwidth]{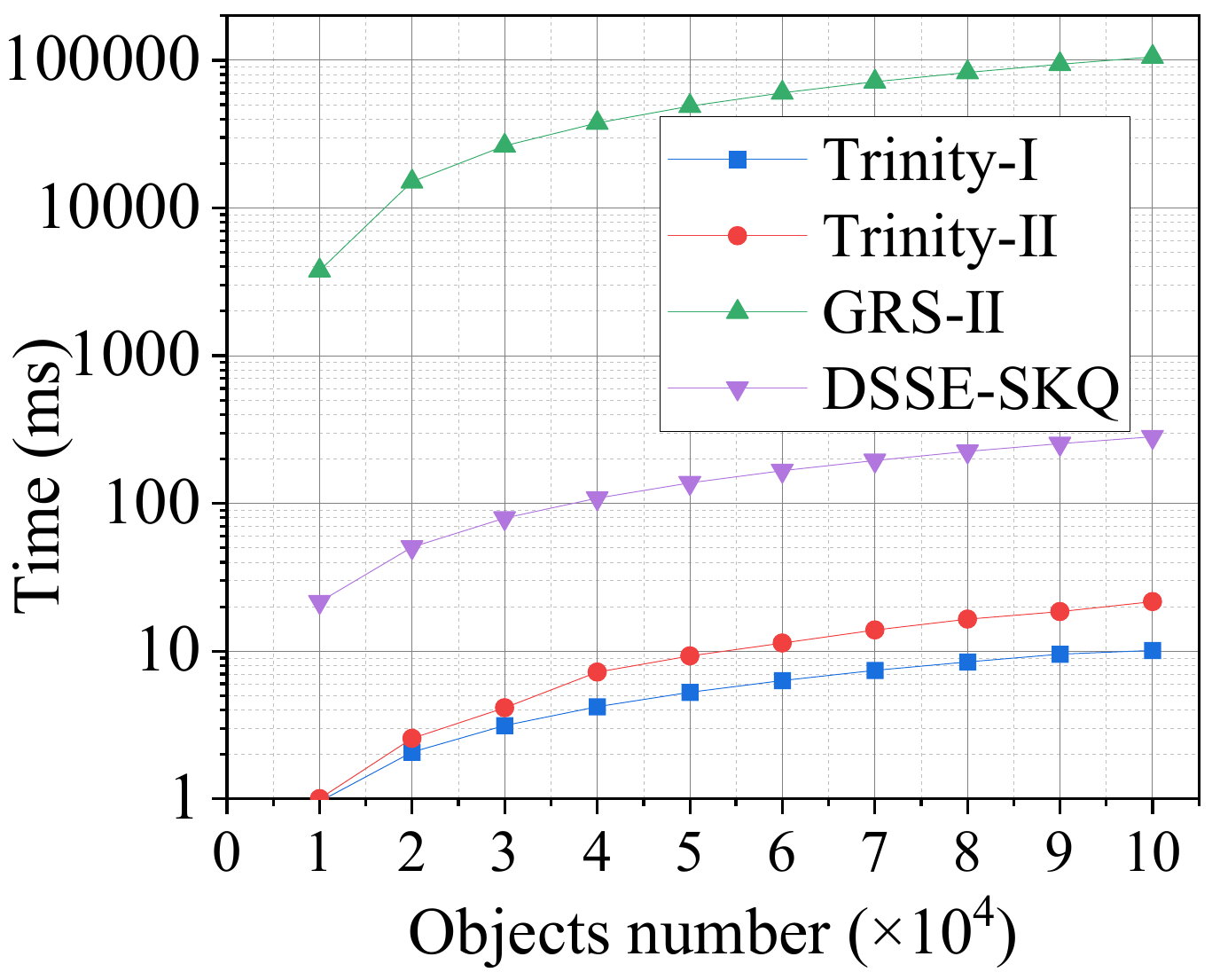}\label{fig:5_update} }
\subfloat[Deletion performance]{\includegraphics[width=0.50\columnwidth]{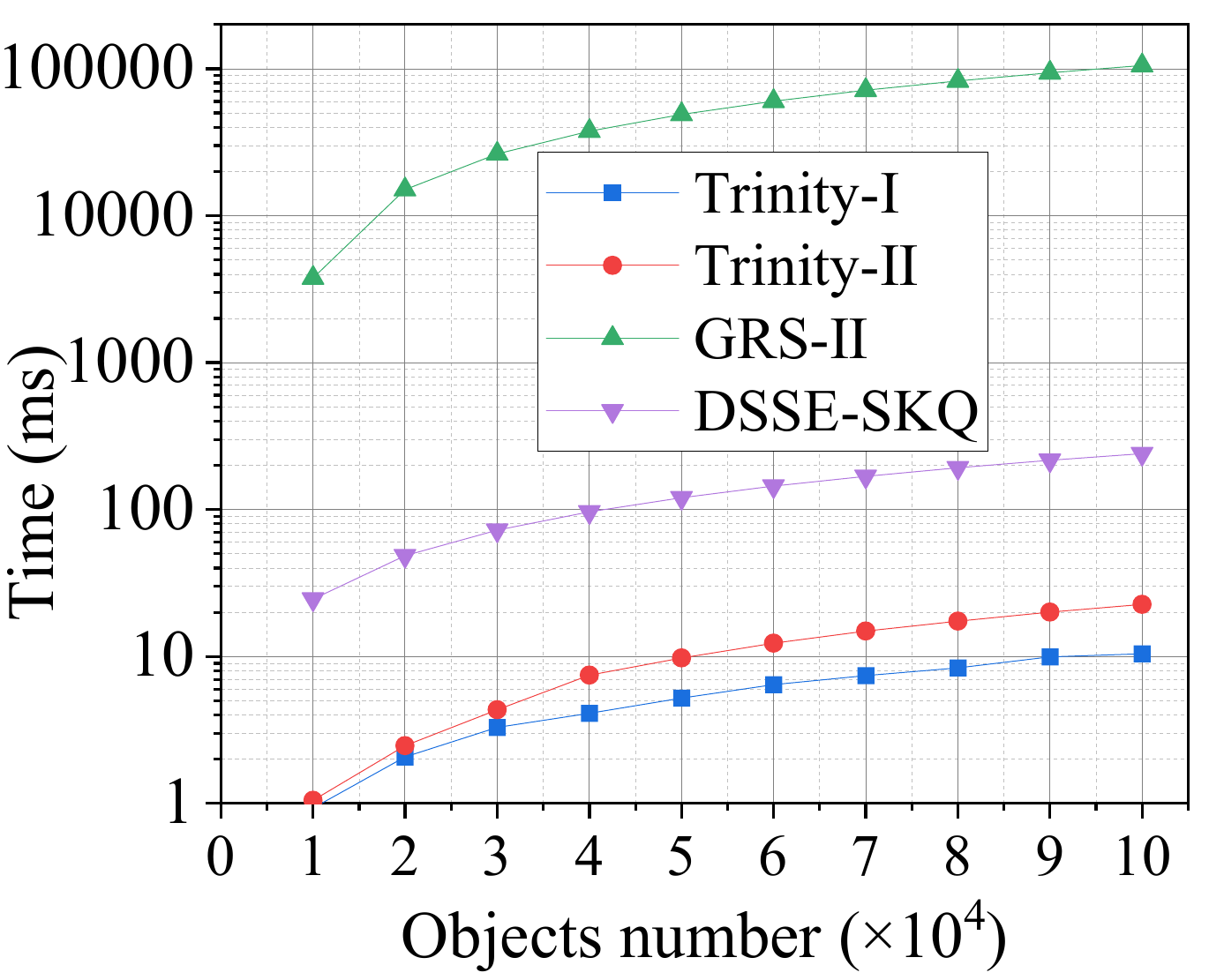}\label{fig:5_delete} }
}
\caption{Trinity-\uppercase\expandafter{\romannumeral 1} vs Trinity-\uppercase\expandafter{\romannumeral 2} , GRS-\uppercase\expandafter{\romannumeral 2} and $\mathsf {DSSE}_{\mathsf {SKQ}}$ update performance}
\label{fig:5_com} 
\end{figure}
\subsection{Update Performance}

\par \textbf{Addition performance.}  The addition time for Trinity-\uppercase\expandafter{\romannumeral 1}, Trinity-\uppercase\expandafter{\romannumeral 2}  and $\mathsf {DSSE}_{\mathsf {SKQ}}$ increases linearly, while for GRS-\uppercase\expandafter{\romannumeral 2} the addition time increases as a logarithmic function. However, as shown in Figure~\ref{fig:5_com}\subref{fig:5_update}, Trinity-\uppercase\expandafter{\romannumeral 1}, Trinity-\uppercase\expandafter{\romannumeral 2}  and $\mathsf {DSSE}_{\mathsf {SKQ}}$  at 100,000 points have an addition time of 10.09 ms, 21.59 ms and 282.77 ms respectively, which is negligible compared with the 105,301 ms of GRS-\uppercase\expandafter{\romannumeral 2}. This is because each addition in GRS-\uppercase\expandafter{\romannumeral 2} is also an update of the entire binary tree. Also, SKSE-\uppercase\expandafter{\romannumeral 2} will not be discussed here because it does not support dynamic updates.

\par \textbf{Deletion performance.} Similar to the add operation, the deletion time for Trinity-\uppercase\expandafter{\romannumeral 1}, Trinity-\uppercase\expandafter{\romannumeral 2}  and $\mathsf {DSSE}_{\mathsf {SKQ}}$ increases linearly, while for for GRS-\uppercase \expandafter{\romannumeral 2}, the deletion time increases as logarithmically. As shown in Figure~\ref{fig:5_com}\subref{fig:5_delete}, Trinity-\uppercase\expandafter{\romannumeral 1}, Trinity-\uppercase\expandafter{\romannumeral 2}  and $\mathsf {DSSE}_{\mathsf {SKQ}}$ have deletion times of 10.44 ms, 22.91 ms and 243.53 ms at 100,000 points, respectively, which are negligible compared to the 105,277 ms of GRS-\uppercase\expandafter{\romannumeral 2}. This is because each deletion in GRS-\uppercase\expandafter{\romannumeral 2} is also an update of the entire binary tree. Again, SKSE-\uppercase\expandafter{\romannumeral 2} will not be discussed here because it does not support dynamic updates.

\section{Conclusion}\label{sec:conclusion}

We propose a novel spatio-temporal data DSSE scheme \TrinityI that supports efficient dynamic updates, and automatic scalability. \TrinityI utilizes Hilbert curves and quotient filters to achieve spatio-temporal range query, implementing IND-SCPA security based on SHVE technology. Additionally, \TrinityI can perform millions of data retrievals within just a few milliseconds. Our solution simultaneously addresses the issues of low query efficiency, scalability challenges, and lack of deletion support, making it more suitable for tasks in spatio-temporal scenarios. But \TrinityI is storage expensive and vulnerable to file injection attacks. So we propose a \TrinityII to solve those questions. Our \TrinityII is forward-secure and storage-saving. Our \TrinityII saves 80\% storage cost than \TrinityI, and eliminate false positive by verification. However, the reduced storage cost also results in a smaller capacity for the QF, making hash collisions more likely. Consequently, compared to \TrinityI, \TrinityII exhibits increased query latency and update latency. However, experimental data consistently demonstrate that our proposed solution \TrinityII, outperforms existing solutions in terms of query efficiency, update efficiency, and storage cost. 

\TrinityII only provides forward security to protect against attack from \cite{zhang2016all,cash2015leakage} 
 caused by the add operation. While the leakage caused by the delete operation is not taken into account, and the backward security for the spatio-temporal data security retrieval scheme requires further improvement.
\section*{Acknowledgements}

This study was partially supported by the National Key R\&D Program of China (No.2022YFB4501000), the National Natural Science Foundation of China (No.62232010, 62302266, 62202364, U23A20302, U24A20244), Shandong Science Fund for Excellent Young Scholars (No.2023HWYQ-008), and Shandong Science Fund for Key Fundamental Research Project (ZR2022ZD02), the fellowship of China National Postdoctoral Program for Innovation Talents (No. BX20230279), the China Postdoctoral Science Foundation (No. 2024M752534), and the Key Research and Development Program of Shaanxi (No.2024GX-YBXM-075).

\ifCLASSOPTIONcaptionsoff
  \newpage
\fi

\begin{IEEEbiography}[{\includegraphics[width=1in,height=1.25in,clip,keepaspectratio]{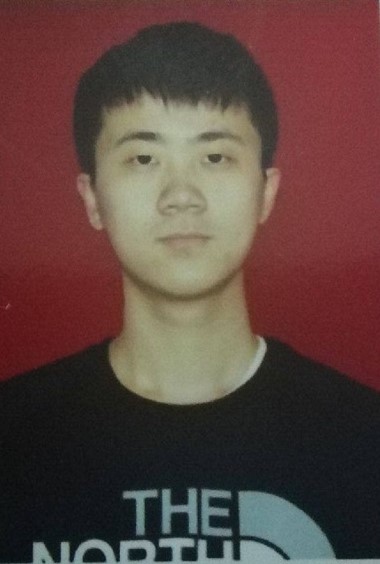}}]{Zhijun Li}
	received the B.E. degree in school of software engineering from Dalian University of Technology, Dalian, China, in 2018, and the Ph.D degree in Cyberspace security from Xidian University, Xi'an, China, in 2024. He is currently a postdoctoral researcher with the School of Computer Science and Technology, Shandong University, China. His current research interests include cloud computing security, edge computing security, and applied cryptography.
\end{IEEEbiography}
\vspace{-4em}
\begin{IEEEbiography}[{\includegraphics[width=1in,height=1.25in,clip,keepaspectratio]{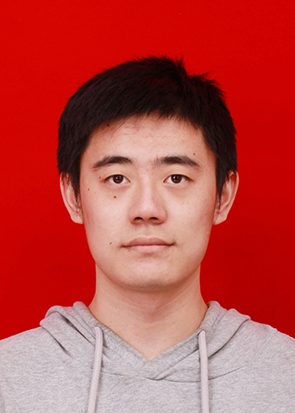}}]{Kuizhi Liu} received the B.S. degree from Nankai University, Tianjin, China, in 2021. He is currently a Ph.D. candidate in School of Cyber Engineering, Xidian University. His current research interests include cloud computing security, database security and data applied cryptography.
\end{IEEEbiography}
\vspace{-4em}
\begin{IEEEbiography}[{\includegraphics[width=1in,height=1.25in,clip,keepaspectratio]{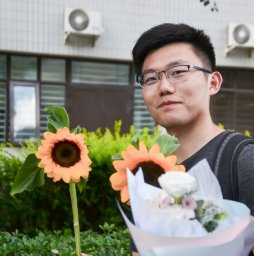}}]	{Minghui Xu (Member, IEEE)}
 received the BS degree in physics from Beijing Normal University, Beijing, China, in 2018, and the PhD degree in computer science from George Washington University, Washington DC, USA, in 2021. He is currently an associate professor with the School of Computer Science and Technology, Shandong University, China. His research focuses on blockchain, distributed computing, and applied cryptography.
\end{IEEEbiography}
\vspace{-4em}
\begin{IEEEbiography}[{\includegraphics[width=1in,height=1.25in,clip,keepaspectratio]{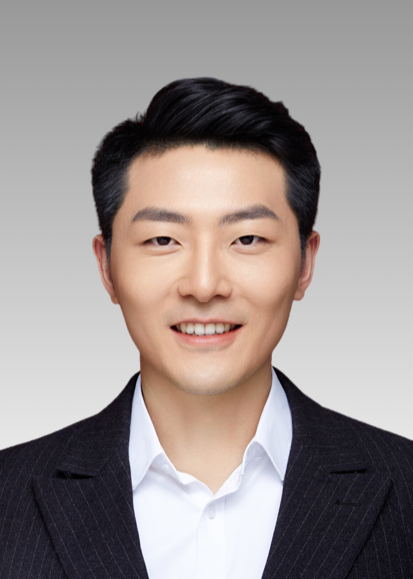}}]	{Xiangyu Wang}
received the B.E. degree and the Ph.D. degree from Xidian University, Xi'an, China, in 2017 and 2021, respectively. He is currently an associate professor with School of Cyber Engineering, Xidian University, Xi'an, China. He received the Outstanding Doctoral Dissertation Award from the China Institute of Communications in 2022. His research interests include big data security, cloud security, and applied cryptography.
\end{IEEEbiography}
\vspace{-4em}
\begin{IEEEbiography}[{\includegraphics[width=1in,height=1.25in,clip,keepaspectratio]{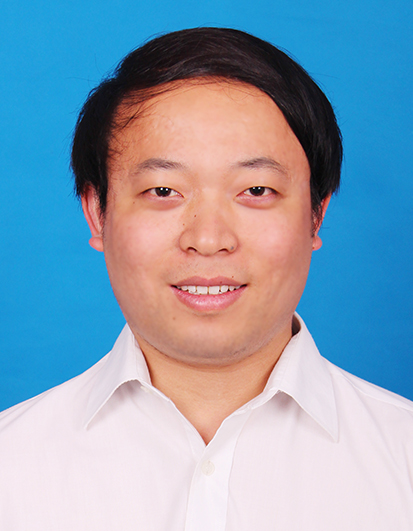}}]{Yinbin Miao}
	(Member, IEEE) received the Ph.D. degree from the Department of Telecommunication Engineering, Xidian University, Xi’an, China, in 2016. He was a Post-Doctoral Fellow with Nanyang Technological University, Singapore, from 2018 to 2019. He is currently a Lecturer with the Department of Cyber Engineering, Xidian University. His current research interests include information security and applied cryptography.
\end{IEEEbiography}
\vspace{-4em}
\begin{IEEEbiography}[{\includegraphics[width=1in,height=1.25in,clip,keepaspectratio]{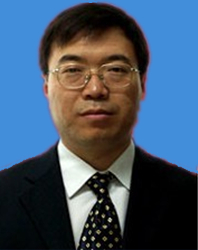}}]{Jianfeng Ma}
	(Member, IEEE) received the Ph.D. degree in computer software and telecommunication engineering from Xidian University, Xi’an, China, in 1995. He was a Research Fellow with Nanyang Technological University, Singapore, from 1999 to 2001. He is currently a Professor and a Ph.D. Supervisor with the Department of Cyber Engineering, Xidian University. His current research interests include information and network security, wireless and mobile computing systems, and computer networks.
\end{IEEEbiography}
\vspace{-4em}
\begin{IEEEbiography}[{\includegraphics[width=1in,height=1.25in,clip,keepaspectratio]{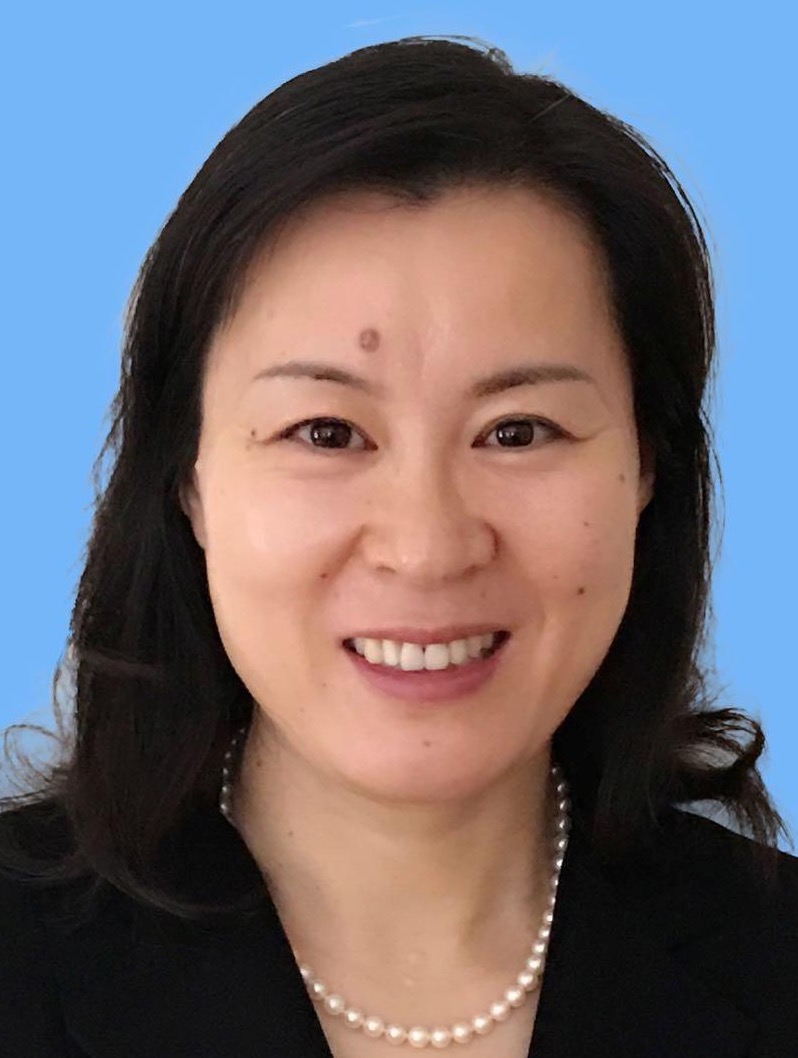}}]{Xiuzhen Cheng}
	(Fellow, IEEE) Xiuzhen Cheng (Fellow, IEEE) received the MS and PhD degrees in computer science from the University of Minnesota – Twin Cities, in 2000 and 2002, respectively. She is a professor with the School of Computer Science and Technology, Shandong University. Her current research interests include wireless and mobile security, cyber physical systems, wireless and mobile computing, sensor networking, and algorithm design and analysis. She has served on the editorial boards of several technical journals and the technical program committees of various professional conferences/workshops. She also has chaired several international conferences. She worked as a program director for the US National Science Foundation (NSF) from April to October in 2006 (full time), and from April 2008 to May 2010 (part time). She received the NSF CAREER Award in 2004. She is a member of ACM.
\end{IEEEbiography}
\vspace{-4em}

%



\end{document}